\documentclass[twoside,10pt]{article}
\usepackage{CJK}
\usepackage{indentfirst}
\usepackage{bm}

\usepackage[numbers,sort&compress]{natbib}
 \usepackage[medium,compact]{titlesec}
\usepackage{float}
\usepackage{color}
\usepackage{graphicx}
\usepackage{subfig}
\usepackage{amsmath}
\usepackage{mathrsfs}
\usepackage{amssymb}
\usepackage{amsthm}
\newtheorem{theorem}{Theorem}
\newtheorem{lemma}{Lemma}

\newtheorem{definition}{Definition}
\newtheorem{proposition}{Proposition}


\begin{document}
\begin{CJK*}{GBK}{song}
\allowdisplaybreaks


\begin{center}
\LARGE\bf
A class of extended high-dimensional nonisospectral KdV hierarchies and symmetry
\end{center}
\footnotetext{\hspace*{-.45cm}\footnotesize $^\ast$Corresponding author: Y.F. Zhang. E-mail: zhangyfcumt@163.com}
\begin{center}
\ \\Haifeng Wang$^{1}$, Yufeng Zhang$^{2,\ast}$, Binlu Feng$^{2}$
\end{center}
\begin{center}
\begin{small} \sl
{$^{1}$School of Science, Jimei University, Xiamen, Fujian, 361021, China\\}
{$^2$School of Mathematics and Information Sciences, Weifang University, Weifang, 261061, China}
\end{small}
\end{center}

\vspace*{2mm}
\begin{center}
\begin{minipage}{15.5cm}
\parindent 20pt\footnotesize

 \noindent {\bfseries
Abstract}: We construct a new class of $N$-dimensional Lie algebras and apply them to integrable systems.
In this paper, we  obtain a nonisospectral KdV integrable hierarchy by introducing a nonisospectral spectral problem.
Then,  a coupled nonisospectral KdV hierarchy  is deduced by means of the corresponding higher-dimensional loop algebra. It follows that the $K$ symmetries, $\tau$ symmetries and their Lie algebra of the coupled nonisospectral KdV hierarchy are investigated.  The  bi-Hamiltonian structures of the both resulting hierarchies  are derived by using the trace identity.  Finally, we derive a multi-component nonisospectral KdV hierarchy related to the $N$-dimensional loop algebra, which   generalizes the coupled results to an arbitrary number of components.

\end{minipage}
\end{center}
\begin{center}
\begin{minipage}{15.5cm}
\begin{minipage}[t]{2.3cm}{\bf Key words:}\end{minipage}
\begin{minipage}[t]{13.1cm}
 Multi-component nonisospectral  KdV hierarchies;  High-dimensional Lie algebras; Bi-Hamiltonian structures;
 Symmetries
\end{minipage}\par\vglue8pt
{\bf MSC Numbers:} 37K05, 37K40, 35Q53 \vskip 0.5cm
\vskip 0.5cm
\end{minipage}
\end{center}
\def\cdot{{\scriptstyle\,\bullet\,}}
\section{Introduction}  
A most active field of recent research involved in applied mathematics and theoretical physics is concerned with nonlinear partial differential equation \cite{3}. The Korteweg-de Vries (KdV) equation \cite{220}, arises in the description of many phenomena in physics \cite{221,2}, is one of the most famous nonlinear partial differential equations. In recent years, many work has been done on the research of isospectral KdV hierarchies  \cite{222,223,224}.
However,  there are few results on the extended high-dimensional nonisospectral KdV hierarchies. Therefore, it is meaningful for us  to consider the generation and application of the extended high-dimensional nonisospectral KdV hierarchies in this paper.
The derivation of integrable hierarchies of evolution equations involves a variety of powerful methods, including the Lax pair method proposed by Magri \cite{1}, the method that Qiao and Ma came up with to generate isospectral and nonisospectral integrable hierarchies by applying the generalized Lax representations \cite{5,6,7,8}, the approach put forward by Tu  \cite{4} and later called Tu-scheme \cite{11}.  Some integrable systems and the corresponding Hamiltonian structures as well as other properties were obtained  by using the Tu scheme, such as  the works in \cite{32,22,23}. Based on it, many extended coupled integrable hierarchies were deduced by means of the knowledge of integrable coupling \cite{30,42,43,103}. The nonisospectral hierarchy is a special case of generalized structure of Lax representations  \cite{aaaa}.  In  \cite{bbbb}, the L-A-B representation of nonlinear evolution equations had been discussed and its range can be determined by an approach  proposed by Qiao and  Strampp. Then, the authors dealt with the category of nonlinear evolution equations in \cite{cccc}. Meanwhile, they proposed an approach for constructing the algebraic structure and $r$-matrix of  nonlinear evolution equations.
However,  the integrable systems generated by the Tu scheme were usually presented under the case of isospectral problems. Also,  to the best of our knowledge, there is very little work on generating multi-component integrable hierarchies because it is extremely dense. Recently, we  constructed a multi-component non-semisimple Lie algebra for generating  higher-dimensional isospectral and nonisospectral integrable hierarchies, and then derived  the $Z_N^{\varepsilon}$ isospectral MKdV integrable coupling hierarchy and the $Z_N^{\varepsilon}$ nonisospectral ANKS integrable coupling hierarchy \cite{206}.

Inspired by the corresponding research results related to the  Frobenius algebra \cite{201,202} and non-semisimple Lie algebra \cite{300,301,205}, we  construct a  class of   higher-dimensional Lie algebras   to generate multi-component  hierarchy of soliton equations. As an  application, we  consider the nonisospectral KdV   spatial spectral problem  under the assumption case where $\lambda_t=\sum\limits_{j=0}^{n}k_j(t)\lambda^{-j}$ \cite{14,21}. It follows that many isospectral and nonisospectral integrable systems can be obtained by reducing these resulting hierarchies.
Actually,  these   nonisospectral integrable systems that we obtained can enrich the existing integrable models and possibly describe new nonlinear phenomena \cite{16,17,18,19,20,26}.

It is known  that many integrable evolution equations possess a new set of symmetries, usually called $\tau$ symmetries, and these symmetries often constitute a Lie algebra together with the original symmetries, called $K$ symmetries \cite{15,207,208,209}. It is known that $\tau$  symmetries  are full of deep necessary \cite{210,211,213,214}. However, to the best of our knowledge,  there are   very few results on the symmetries of  coupled nonisospectral integrable hierarchies.  Therefore, it is   significant for us to investigate the $K$ symmetries, $\tau$ symmetries and their Lie algebra of the coupled nonisospectral KdV hierarchy.

\section{A few expanding higher-dimensional Lie algebras}
 The Lie algebra $A_1$ admits two basic subalgebras \cite{4}. One is that
\begin{equation}\label{1}
h=\left(\begin{matrix}
1&0\\
0&-1\end{matrix}
\right),\ \
e=\left(\begin{matrix}
0&1\\
0&0\end{matrix}
\right),\ \
f=\left(\begin{matrix}
0&0\\
1&0\end{matrix}
\right),
\end{equation}
 which satisfies the  commutative relations
\begin{equation}\notag
[h,e]=2e, \ \ [h,f]=-2f,  \ \ [e,f]=h.
\end{equation}
The other one is that
\begin{equation}\label{2}
\overline{h}=\frac{1}{2}\left(\begin{matrix}
1&0\\
0&-1\end{matrix}
\right),\ \
\overline{e}=\frac{1}{2}\left(\begin{matrix}
0&1\\
1&0\end{matrix}
\right),\ \
\overline{f}=\frac{1}{2}\left(\begin{matrix}
0&1\\
-1&0\end{matrix}
\right),
\end{equation}
 which is equipped with
\begin{equation}\notag
[\overline{h},\overline{e}]=\overline{f}, \ \ [\overline{h},\overline{f}]=\overline{e},  \ \ [\overline{e},\overline{f}]=\overline{h}.
\end{equation}
In \cite{41}, the authors presented several finite-dimensional Lie algebras. Now, we construct a few new higher-dimensional Lie algebras and generalize them to infinite dimensions. For the first subalgebra \eqref{1}, we introduce the extended Lie algebras as follows: \\
$\mathbf{Case\ 1}$
\begin{equation}\label{3}
 A_{12}=\rm span\{h_i\}_{i=1}^6,
 \end{equation}
 where
\begin{equation}\notag
\begin{split}
&h_1=\left(\begin{matrix}
h&0\\
0&h\end{matrix}
\right)=\left(\begin{matrix}
1&0&0&0\\
0&-1&0&0\\
0&0&1&0\\
0&0&0&-1
\end{matrix}
\right),\
h_2=\left(\begin{matrix}
e&0\\
0&e\end{matrix}
\right)=\left(\begin{matrix}
0&1&0&0\\
0&0&0&0\\
0&0&0&1\\
0&0&0&0
\end{matrix}
\right),\
h_3=\left(\begin{matrix}
f&0\\
0&f\end{matrix}
\right)=\left(\begin{matrix}
0&0&0&0\\
1&0&0&0\\
0&0&0&0\\
0&0&1&0
\end{matrix}
\right),\\
&h_4=\left(\begin{matrix}
0&\varepsilon h\\
h&0\end{matrix}
\right)=\left(\begin{matrix}
0&0&\varepsilon&0\\
0&0&0&-\varepsilon\\
1&0&0&0\\
0&-1&0&0
\end{matrix}
\right),\
h_5=\left(\begin{matrix}
0&\varepsilon e\\
e&0\end{matrix}
\right)=\left(\begin{matrix}
0&0&0&\varepsilon\\
0&0&0&0\\
0&1&0&0\\
0&0&0&0
\end{matrix}
\right),\
h_6=\left(\begin{matrix}
0&\varepsilon f\\
f&0\end{matrix}
\right)=\left(\begin{matrix}
0&0&0&0\\
0&0&\varepsilon&0\\
0&0&0&0\\
1&0&0&0
\end{matrix}
\right),
\end{split}
\end{equation}
\begin{equation}\notag
\begin{split}
&[h_1,h_2]=2h_2, \ \ [h_1,h_3]=-2h_3,  \ \ [h_1,h_4]=0, \ \ [h_1,h_5]=2h_5,\ \ [h_1,h_6]=-2h_6,   \\
& [h_2,h_3]=h_1,\ \ [h_2,h_4]=-2h_5,\ \ [h_2,h_5]=0, \ \ [h_2,h_6]=h_4,\ \ [h_3,h_4]=2h_6,  \ \ \\
&[h_3,h_5]=-h_4,\ \ [h_3,h_6]=0,\ \ [h_4,h_5]=2\varepsilon h_2, \ \ [h_4,h_6]=-2\varepsilon h_3, \ \ [h_5,h_6]=\varepsilon h_1,
\end{split}
\end{equation}
with $\varepsilon\in \mathbb{R}$.\\
Let $G_1$=$\rm span\{h_1, h_2, h_3\}$, $G_2$=$\rm span\{h_4, h_5, h_6\}$, then $A_{12}=G_1\bigoplus G_2$.  Denoting \begin{equation}\notag
[G_i,G_j]=\{[A,B]|A\in G_i,B\in G_j\},
\end{equation}
 we find that the  closure  properties  between $G_1$ and $G_2$ are as follows:
\begin{equation}\notag
[G_1,G_1]\subseteq G_1,\ \ [G_1,G_2]\subseteq G_2,\ \ [G_2,G_2]\subseteq G_1.
\end{equation}
$\mathbf{Case\ 2}$
\begin{equation}\label{4}
 A_{13}=\rm span\{\overline{h}_i\}_{i=1}^9,
 \end{equation}
 where
\begin{equation}\notag
\begin{split}
&\overline{h}_1=\left(\begin{matrix}
h&0&0\\
0&h&0\\
0&0&h
\end{matrix}
\right),\ \
\overline{h}_2=\left(\begin{matrix}
e&0&0\\
0&e&0\\
0&0&e
\end{matrix}
\right),\ \
\overline{h}_3=\left(\begin{matrix}
f&0&0\\
0&f&0\\
0&0&f
\end{matrix}
\right), \\
&\overline{h}_4=\left(\begin{matrix}
0&0&\varepsilon h\\
h&0&0\\
0&h&0
\end{matrix}
\right),\ \
\overline{h}_5=\left(\begin{matrix}
0&0&\varepsilon e\\
e&0&0\\
0&e&0
\end{matrix}
\right),\ \
\overline{h}_6=\left(\begin{matrix}
0&0&\varepsilon f\\
f&0&0\\
0&f&0
\end{matrix}
\right), \\
&\overline{h}_7=\left(\begin{matrix}
0&\varepsilon h&0\\
0&0&\varepsilon h\\
h&0&0
\end{matrix}
\right),\ \
\overline{h}_8=\left(\begin{matrix}
0&\varepsilon e&0\\
0&0&\varepsilon e\\
e&0&0
\end{matrix}
\right),\ \
\overline{h}_9=\left(\begin{matrix}
0&\varepsilon f&0\\
0&0&\varepsilon f\\
f&0&0
\end{matrix}
\right),\ \
\end{split}
\end{equation}
\begin{equation}\notag
\begin{split}
&[\overline{h}_1,\overline{h}_2]=2\overline{h}_2, \ \ [\overline{h}_1,\overline{h}_3]=-2\overline{h}_3,  \ \ [\overline{h}_1,\overline{h}_5]=2\overline{h}_5, \ \ [\overline{h}_1,\overline{h}_6]=-2\overline{h}_6,\ \ [\overline{h}_1,\overline{h}_8]=2\overline{h}_8,\ \    \\
&[\overline{h}_1,\overline{h}_9]=-2\overline{h}_9,\ \ [\overline{h}_2,\overline{h}_3]=\overline{h}_1, \ \ [\overline{h}_2,\overline{h}_4]=-2\overline{h}_5,  \ \ [\overline{h}_2,\overline{h}_6]=\overline{h}_4, \ \ [\overline{h}_2,\overline{h}_7]=-2\overline{h}_8,   \\
& [\overline{h}_2,\overline{h}_9]=\overline{h}_7,\ \ [\overline{h}_3,\overline{h}_4]=2\overline{h}_6,\ \ [\overline{h}_3,\overline{h}_5]=-\overline{h}_4, \ \ [\overline{h}_3,\overline{h}_7]=2\overline{h}_9,  \ \ [\overline{h}_3,\overline{h}_8]=-\overline{h}_7, \ \    \\
&[\overline{h}_4,\overline{h}_5]=2\overline{h}_8,\ \ [\overline{h}_4,\overline{h}_6]=-2\overline{h}_9,\ \ [\overline{h}_4,\overline{h}_8]=2\varepsilon\overline{h}_2,\ \ [\overline{h}_4,\overline{h}_9]=-2\varepsilon\overline{h}_3, \ \ [\overline{h}_5,\overline{h}_6]=\overline{h}_7, \\ &[\overline{h}_5,\overline{h}_7]=-2\varepsilon\overline{h}_2,\ \ [\overline{h}_5,\overline{h}_9]=\varepsilon\overline{h}_1,\ \ [\overline{h}_7,\overline{h}_8]=2\varepsilon\overline{h}_5,\ \ [\overline{h}_7,\overline{h}_9]=-2\varepsilon\overline{h}_6,\ \ [\overline{h}_8,\overline{h}_9]=\varepsilon\overline{h}_4,\\
&[\overline{h}_1,\overline{h}_4]=[\overline{h}_2,\overline{h}_5]=[\overline{h}_3,\overline{h}_6]=[\overline{h}_1,\overline{h}_7]
=[\overline{h}_2,\overline{h}_8]=[\overline{h}_3,\overline{h}_9]=[\overline{h}_4,\overline{h}_7]
=[\overline{h}_5,\overline{h}_8]=[\overline{h}_6,\overline{h}_9]=0.
\end{split}
\end{equation}
Let $\overline{G}_1$=$\rm span\{\overline{h}_1, \overline{h}_2, \overline{h}_3\}$, $\overline{G}_2$=$\rm span\{\overline{h}_4, \overline{h}_5, \overline{h}_6\}$, $\overline{G}_3$=$\rm span\{\overline{h}_7, \overline{h}_8, \overline{h}_9\}$, then $A_{13}=\overline{G}_1\bigoplus \overline{G}_2\bigoplus \overline{G}_3$.  It follows that one has
\begin{equation}\notag
\begin{split}
&[\overline{G}_1,\overline{G}_1]\subseteq \overline{G}_1,\ \ [\overline{G}_1,\overline{G}_2]\subseteq \overline{G}_2,\ \ [\overline{G}_1,\overline{G}_3]\subseteq \overline{G}_3,\\
&[\overline{G}_2,\overline{G}_2]\subseteq \overline{G}_3,\ \ [\overline{G}_2,\overline{G}_3]\subseteq \overline{G}_1,\ \ [\overline{G}_3,\overline{G}_3]\subseteq \overline{G}_2.
\end{split}
\end{equation}

Introducing a $N\times N$ square matrix of the following  form:
\begin{equation}\label{5}
M(A_1,A_2,\cdots,A_N)=\left[
\begin{matrix}
A_1&\varepsilon A_N&\varepsilon A_{N-1}&\cdots&\varepsilon A_4&\varepsilon A_3&\varepsilon A_2\\
A_2&A_1&\varepsilon A_{N}&\cdots&\varepsilon A_5&\varepsilon A_4&\varepsilon A_3\\
A_3&A_2&A_{1}&\cdots&\varepsilon A_6&\varepsilon A_5&\varepsilon A_4\\
\vdots&\vdots&\vdots&\ddots&\vdots&\vdots&\vdots\\
A_{N-2}&A_{N-3}&A_{N-4}&\cdots&A_{1}&\varepsilon A_N&\varepsilon A_{N-1}\\
A_{N-1}&A_{N-2}&A_{N-3}&\cdots&A_{2}&A_{1}&\varepsilon A_N\\
A_{N}&A_{N-1}&A_{N-2}&\cdots&A_{3}&A_{2}&A_{1}
\end{matrix}
\right]=:\left[
\begin{matrix}
A_1,A_2,\cdots,A_N
\end{matrix}
\right]^T,
\end{equation}
where $A_m \ (1\leq m\leq N)$  represent $N$ arbitrary square matrices of the same order. Here we use the vector $\left[
\begin{matrix}
A_1,A_2,\cdots,A_N
\end{matrix}
\right]^T$ to represent the corresponding $N\times N$  matrix for convenience.\\
$\mathbf{Case\ 3}$
\begin{equation}\label{6}
A_{1N}=\rm span\{\widetilde{h}_i\}_{i=1}^{3N},
\end{equation}
with
\begin{equation}\notag
\begin{split}
&\widetilde{h}_1=M(h,0,\cdots,0),\ \ \widetilde{h}_2=M(e,0,\cdots,0),\ \ \widetilde{h}_3=M(f,0,\cdots,0),\\
&\widetilde{h}_4=M(0,h,\cdots,0),\ \ \widetilde{h}_5=M(0,e,\cdots,0),\ \ \widetilde{h}_6=M(0,f,\cdots,0),\\
&\cdots\ \ \\
&\widetilde{h}_{3N-2}=M(0,0,\cdots,h),\ \ \widetilde{h}_{3N-1}=M(0,0,\cdots,e),\ \ \widetilde{h}_{3N}=M(0,0,\cdots,f),\ \ N=1,2,\cdots,
\end{split}
\end{equation}
\begin{equation}\notag
\begin{split}
&[\widetilde{h}_{3i-2},\widetilde{h}_{3k-2}]=[\widetilde{h}_{3i-1},\widetilde{h}_{3k-1}]=[\widetilde{h}_{3i},\widetilde{h}_{3k}]=0,\  \ \ \ i,k=1,2,\cdots,N,\\
&[\widetilde{h}_{3i-2},\widetilde{h}_{3k-1}]=\begin{cases}
  2\widetilde{h}_{3(k+i-1)-1}, \ \ \ 1\leq k\leq N-i+1,\\
 2\varepsilon\widetilde{h}_{3(k+i-1-N)-1}, \ \ \ N-i+2\leq k\leq N,
  \end{cases}\\
 & [\widetilde{h}_{3i-2},\widetilde{h}_{3k}]=\begin{cases}
  -2\widetilde{h}_{3(k+i-1)}, \ \ \ 1\leq k\leq N-i+1,\\
 -2\varepsilon\widetilde{h}_{3(k+i-1-N)}, \ \ \ N-i+2\leq k\leq N,
  \end{cases}\\
  &[\widetilde{h}_{3i-1},\widetilde{h}_{3k}]=\begin{cases}
  \widetilde{h}_{3(k+i-1)-2}, \ \ \ 1\leq k\leq N-i+1,\\
 \varepsilon\widetilde{h}_{3(k+i-1-N)-2}, \ \ \ N-i+2\leq k\leq N,
  \end{cases}
\end{split}
\end{equation}
where $M$ is  $N\times N$ matrix given by \eqref{5} and $h, e, f$ are the second-order matrices given by \eqref{1}.\\
Let $\widetilde{G}_k$=$\rm span\{\widetilde{h}_{3k-2}, \widetilde{h}_{3k-1}, \widetilde{h}_{3k}\}$, $k=1,2,\cdots N$, then $A_{1N}=\widetilde{G}_1\bigoplus \widetilde{G}_2\bigoplus\cdots\bigoplus \widetilde{G}_N$.  Thus, we obtain
\begin{equation}\notag
\begin{split}
&[\widetilde{G}_i,\widetilde{G}_j]\subseteq \widetilde{G}_{i+j-1-\delta N},\ \   \delta=\begin{cases}
  0, \ \ \ 2\leq i+j\leq N+1,\\
  1, \ \ \ N+2\leq i+j\leq 2N,
  \end{cases}\ \  i,j=1,2,\cdots,N.
\end{split}
\end{equation}

Similarly,   we introduce the following  extended Lie algebras related to the subalgebra \eqref{2}:\\
$\mathbf{Case\ 4}$
\begin{equation}\label{7}
 A_{22}=\rm span\{\overline{e}_i\}_{i=1}^6,
 \end{equation}
 where
\begin{equation}\notag
\begin{split}
&\overline{e}_1=\left(\begin{matrix}
\overline{h}&0\\
0&\overline{h}\end{matrix}
\right)=\frac{1}{2}\left(\begin{matrix}
1&0&0&0\\
0&-1&0&0\\
0&0&1&0\\
0&0&0&-1
\end{matrix}
\right),\ \
\overline{e}_2=\left(\begin{matrix}
\overline{e}&0\\
0&\overline{e}\end{matrix}
\right)=\frac{1}{2}\left(\begin{matrix}
0&1&0&0\\
1&0&0&0\\
0&0&0&1\\
0&0&1&0
\end{matrix}
\right),\\
&\overline{e}_3=\left(\begin{matrix}
\overline{f}&0\\
0&\overline{f}\end{matrix}
\right)=\frac{1}{2}\left(\begin{matrix}
0&1&0&0\\
-1&0&0&0\\
0&0&0&1\\
0&0&-1&0
\end{matrix}
\right),\ \
\overline{e}_4=\left(\begin{matrix}
0&\varepsilon \overline{h}\\
\overline{h}&0\end{matrix}
\right)=\frac{1}{2}\left(\begin{matrix}
0&0&\varepsilon&0\\
0&0&0&-\varepsilon\\
1&0&0&0\\
0&-1&0&0
\end{matrix}
\right),\\
&\overline{e}_5=\left(\begin{matrix}
0&\varepsilon \overline{e}\\
\overline{e}&0\end{matrix}
\right)=\frac{1}{2}\left(\begin{matrix}
0&0&0&\varepsilon\\
0&0&\varepsilon&0\\
0&1&0&0\\
1&0&0&0
\end{matrix}
\right),\ \
\overline{e}_6=\left(\begin{matrix}
0&\varepsilon \overline{f}\\
\overline{f}&0\end{matrix}
\right)=\frac{1}{2}\left(\begin{matrix}
0&0&0&\varepsilon\\
0&0&-\varepsilon&0\\
0&1&0&0\\
-1&0&0&0
\end{matrix}
\right),
\end{split}
\end{equation}
\begin{equation}\notag
\begin{split}
&[\overline{e}_1,\overline{e}_2]=\overline{e}_3, \ \ [\overline{e}_1,\overline{e}_3]=\overline{e}_2,  \ \ [\overline{e}_1,\overline{e}_4]=0, \ \ [\overline{e}_1,\overline{e}_5]=\overline{e}_6,\ \ [\overline{e}_1,\overline{e}_6]=\overline{e}_5,   \\
& [\overline{e}_2,\overline{e}_3]=-\overline{e}_1,\ \ [\overline{e}_2,\overline{e}_4]=-\overline{e}_6,\ \ [\overline{e}_2,\overline{e}_5]=0, \ \ [\overline{e}_2,\overline{e}_6]=-\overline{e}_4,\ \ [\overline{e}_3,\overline{e}_4]=-\overline{e}_5,  \ \ \\
&[\overline{e}_3,\overline{e}_5]=\overline{e}_4,\ \ [\overline{e}_3,\overline{e}_6]=0,\ \ [\overline{e}_4,\overline{e}_5]=\varepsilon \overline{e}_3, \ \ [\overline{e}_4,\overline{e}_6]=\varepsilon \overline{e}_2, \ \ [\overline{e}_5,\overline{e}_6]=\varepsilon \overline{e}_1.
\end{split}
\end{equation}
Let $\overline{g}_1$=$\rm span\{\overline{e}_1, \overline{e}_2, \overline{e}_3\}$, $\overline{g}_2$=$\rm span\{\overline{e}_4, \overline{e}_5, \overline{e}_6\}$, then $A_{22}=\overline{g}_1\bigoplus \overline{g}_2$.
It follows that one has
\begin{equation}\notag
[\overline{g}_1,\overline{g}_1]\subseteq \overline{g}_1,\ \ [\overline{g}_1,\overline{g}_2]\subseteq \overline{g}_2,\ \ [\overline{g}_2,\overline{g}_2]\subseteq \overline{g}_1.
\end{equation}
$\mathbf{Case\ 5}$
\begin{equation}\label{8}
A_{2N}=\rm span\{\widetilde{e}_i\}_{i=1}^{3N},
\end{equation}
with
\begin{equation}\notag
\begin{split}
&\widetilde{e}_1=M(\overline{h},0,\cdots,0),\ \ \widetilde{e}_2=M(\overline{e},0,\cdots,0),\ \ \widetilde{e}_3=M(\overline{f},0,\cdots,0),\\
&\widetilde{e}_4=M(0,\overline{h},\cdots,0),\ \ \widetilde{e}_5=M(0,\overline{e},\cdots,0),\ \ \widetilde{e}_6=M(0,\overline{f},\cdots,0),\\
&\cdots\ \ \\
&\widetilde{e}_{3N-2}=M(0,0,\cdots,\overline{h}),\ \ \widetilde{e}_{3N-1}=M(0,0,\cdots,\overline{e}),\ \ \widetilde{e}_{3N}=M(0,0,\cdots,\overline{f}),\ \ N=1,2,\cdots,
\end{split}
\end{equation}
\begin{equation}\notag
\begin{split}
&[\widetilde{e}_{3i-2},\widetilde{e}_{3k-2}]=[\widetilde{e}_{3i-1},\widetilde{e}_{3k-1}]=[\widetilde{e}_{3i},\widetilde{e}_{3k}]=0,\ \ \ \ i,k=1,2,\cdots,N,\\
&[\widetilde{e}_{3i-2},\widetilde{e}_{3k-1}]=\begin{cases}
  \widetilde{e}_{3(k+i-1)}, \ \ \ 1\leq k\leq N-i+1,\\
 \varepsilon\widetilde{e}_{3(k+i-1-N)}, \ \ \ N-i+2\leq k\leq N,
  \end{cases}\\
 & [\widetilde{e}_{3i-2},\widetilde{e}_{3k}]=\begin{cases}
  \widetilde{e}_{3(k+i-1)-1}, \ \ \ 1\leq k\leq N-i+1,\\
 \varepsilon\widetilde{e}_{3(k+i-1-N)-1}, \ \ \ N-i+2\leq k\leq N,
  \end{cases}\\
  &[\widetilde{e}_{3i-1},\widetilde{e}_{3k}]=\begin{cases}
  -\widetilde{e}_{3(k+i-1)-2}, \ \ \ 1\leq k\leq N-i+1,\\
 -\varepsilon\widetilde{h}_{3(k+i-1-N)-2}, \ \ \ N-i+2\leq k\leq N,
  \end{cases}
\end{split}
\end{equation}
where $M$ is  $N\times N$ matrix given by \eqref{5} and $\overline{h}, \overline{e}, \overline{f}$ are  given by \eqref{2}.\\
Let $\widetilde{g}_k$=$\rm span\{\widetilde{e}_{3k-2}, \widetilde{e}_{3k-1}, \widetilde{e}_{3k}\}$, $k=1,2,\cdots N$, then $A_{2N}=\widetilde{g}_1\bigoplus \widetilde{g}_2\bigoplus\cdots\bigoplus \widetilde{g}_N$.  Thus, we obtain
\begin{equation}\notag
\begin{split}
&[\widetilde{g}_i,\widetilde{g}_j]\subseteq \widetilde{g}_{i+j-1-\delta N},\ \   \delta=\begin{cases}
  0, \ \ \ 2\leq i+j\leq N+1,\\
  1, \ \ \ N+2\leq i+j\leq 2N,
  \end{cases}\ \  i,j=1,2,\cdots,N.
\end{split}
\end{equation}

We believe that other Lie algebras could also be extended to
the similar higher-dimensional forms. For example, the   Lie algebra $\rm so(3)$ admits a set of bases $f_1, f_2, f_3$,
 \begin{equation}\label{9}
 f_1=\left(\begin{matrix}
0&0&1\\
0&0&0\\
-1&0&0
\end{matrix}
\right),\ \
 f_2=\left(\begin{matrix}
0&0&0\\
0&0&-1\\
0&1&0
\end{matrix}
\right),\ \
 f_3=\left(\begin{matrix}
0&1&0\\
-1&0&0\\
0&0&0
\end{matrix}
\right),\ \
 \end{equation}
 whose commutators are
\begin{equation}\notag
[f_1,f_2]=f_3, \ \ [f_1,f_3]=-f_2,  \ \ [f_2,f_3]=f_1.
\end{equation}
The Lie algebra $\rm so(3)$ can be extended into:\\
$\mathbf{Case\ 6}$
\begin{equation}\label{10}
 A_{32}=\rm span\{\overline{f}_i\}_{i=1}^6,
 \end{equation}
 where
\begin{equation}\notag
\begin{split}
&\overline{f}_1=\left(\begin{matrix}
f_1&0\\
0&f_1\end{matrix}
\right),\ \
\overline{f}_2=\left(\begin{matrix}
f_2&0\\
0&f_2\end{matrix}
\right),\ \
\overline{f}_3=\left(\begin{matrix}
f_3&0\\
0&f_3\end{matrix}
\right),\\
&\overline{f}_4=\left(\begin{matrix}
0&\varepsilon f_1\\
f_1&0\end{matrix}
\right),\ \
\overline{f}_5=\left(\begin{matrix}
0&\varepsilon f_2\\
f_2&0\end{matrix}
\right),\ \ \overline{f}_6=\left(\begin{matrix}
0&\varepsilon f_3\\
f_3&0\end{matrix}
\right),
\end{split}
\end{equation}
\begin{equation}\notag
\begin{split}
&[\overline{f}_1,\overline{f}_2]=\overline{f}_3, \ \ [\overline{f}_1,\overline{f}_3]=-\overline{f}_2,  \ \ [\overline{f}_1,\overline{f}_4]=0, \ \ [\overline{f}_1,\overline{f}_5]=\overline{f}_6,\ \ [\overline{f}_1,\overline{f}_6]=-\overline{f}_5,   \\
& [\overline{f}_2,\overline{f}_3]=\overline{f}_1,\ \ [\overline{f}_2,\overline{f}_4]=-\overline{f}_6,\ \ [\overline{f}_2,\overline{f}_5]=0, \ \ [\overline{f}_2,\overline{f}_6]=\overline{f}_4,\ \ [\overline{f}_3,\overline{f}_4]=\overline{f}_5,  \ \ \\
&[\overline{f}_3,\overline{f}_5]=-\overline{f}_4,\ \ [\overline{f}_3,\overline{f}_6]=0,\ \ [\overline{f}_4,\overline{f}_5]=\varepsilon \overline{f}_3, \ \ [\overline{f}_4,\overline{f}_6]=-\varepsilon \overline{f}_2, \ \ [\overline{f}_5,\overline{f}_6]=\varepsilon \overline{f}_1.
\end{split}
\end{equation}
$\mathbf{Case\ 7}$
\begin{equation}\label{11}
A_{3N}=\rm span\{\widetilde{f}_i\}_{i=1}^{3N},
\end{equation}
with
\begin{equation}\notag
\begin{split}
&\widetilde{f}_1=M(f_1,0,\cdots,0),\ \ \widetilde{f}_2=M(f_2,0,\cdots,0),\ \ \widetilde{f}_3=M(f_3,0,\cdots,0),\\
&\widetilde{f}_4=M(0,f_1,\cdots,0),\ \ \widetilde{f}_5=M(0,f_2,\cdots,0),\ \ \widetilde{f}_6=M(0,f_3,\cdots,0),\\
&\cdots\ \ \\
&\widetilde{f}_{3N-2}=M(0,0,\cdots,f_1),\ \ \widetilde{f}_{3N-1}=M(0,0,\cdots,f_2),\ \ \widetilde{f}_{3N}=M(0,0,\cdots,f_3),\ \ N=1,2,\cdots,
\end{split}
\end{equation}
where $M$ is  given by \eqref{5} and $f_1, f_2, f_3$ are  given by \eqref{9}.

These expanded higher-dimensional Lie algebras can be applied to different spectral problems, and then generate multi-component integrable hierarchies. Below we consider a specific application, that is the KdV spectral problem.

\section{A nonisospectral KdV hierarchy}

For the first set of basis  \eqref{1}, we consider the corresponding loop algebra
\begin{equation}\notag
\widetilde{A}_{11}=\rm span\{h(n), e(n), f(n)\}
\end{equation}
which satisfies the commutators
\begin{equation}\notag
[h(n),e(m)]=2e(m+n),\ [h(n),f(m)]=-2f(m+n),\ [e(n),f(m)]=h(m+n), \ m,n\in \mathbb{Z},
\end{equation}
where
\begin{equation}\notag
h(n)=h\lambda^{n}, \ e(n)=e\lambda^{n}, \ f(n)=f\lambda^{n}.
\end{equation}
Based on the     KdV spatial  spectral problem  (see \cite{2}), we introduce the following  nonisospectral problem based on $\widetilde{A}_{11}$
\begin{equation}\label{12}
\begin{cases}
\psi_x=M\psi,\ \ M=\frac{1}{4}e(1)-ue(0)+f(0)=\left(\begin{matrix}
0&-u+\frac{\lambda}{4}\\
1&0
\end{matrix}
\right), \\
\psi_t=N\psi,\ \ N=ah(0)+be(0)+cf(0)=\left(\begin{matrix}
a&b\\
c&-a\end{matrix}
\right),\\
\lambda_t=\sum\limits_{m\geq0}k_m(t)\lambda^{-m},
\end{cases}
\end{equation}
where  $a=\sum\limits_{m\geq0}a_m\lambda^{-m}$, $b=\sum\limits_{m\geq0}b_m\lambda^{-m}$, $c=\sum\limits_{m\geq0}c_m\lambda^{-m}$.\\
A direct calculation  gives
$\frac{\partial M}{\partial\lambda}\lambda_t=\frac{1}{4}\sum\limits_{m\geq0}k_m(t)e(-m).$
 By  solving the following nonisospectral stationary zero-curvature equation of \eqref{12}
\begin{equation}\label{13}
N_x=\frac{\partial M}{\partial \lambda}\lambda_t+[M,N],
\end{equation}
 we obtain the  recursion equations:
\begin{equation}\label{14}
\begin{cases}
a_{mx}=-uc_{m}+\frac{1}{4}c_{m+1}-b_m,\\
b_{mx}=\frac{1}{4}k_m(t)-\frac{1}{2}a_{m+1}+2ua_{m},\\
c_{mx}=2a_m,
\end{cases}
\end{equation}
which have the  equivalent forms
\begin{equation}\label{15}
\begin{cases}
a_{m}=\frac{1}{2}c_{mx},\\
b_m=-\frac{1}{4}c_{m+1}+\partial^{-1}uc_{mx}+\frac{1}{4}k_m(t)x,\\
c_{m+1}=(\partial^2+2u+2\partial^{-1}u\partial)c_m+\frac{1}{2}k_m(t)x.
\end{cases}
\end{equation}
To the  recursion equations \eqref{15}, we take the initial values $ a_0=0,$ then
\begin{equation}\label{15*}
\begin{split}
& a_0=0,\ \ c_0=\alpha,\ \ c_1=2\alpha u+\frac{1}{2}k_0(t)x,\ \ b_0=-\frac{1}{2}\alpha u+\frac{1}{8}k_0(t)x,\\
&a_1=\alpha u_x+\frac{1}{4}k_0(t),\ \ c_2=2\alpha u_{xx}+6\alpha u^2+k_0(t)(xu+\partial^{-1}u)+\frac{1}{2}k_1(t)x,\\ &b_1=-\frac{\alpha}{2}u_{xx}-\frac{\alpha}{2}u^2-\frac{1}{4}k_0(t)(xu+\partial^{-1}u)-\frac{1}{8}k_1(t)x+\frac{1}{4}k_2(t)x,\\
&\cdots
\end{split}
\end{equation}

Denoting that
\begin{equation}\notag
N^{(n)}=N_+^{(n)}+N_-^{(n)}=N\lambda^n,\ \
N_+^{(n)}=\sum\limits_{i=0}^n(a_ih(n-i)+b_ie(n-i)+c_if(n-i)),
\end{equation}
\begin{equation}\notag
\deg h(n)=\deg(h\lambda^{n})=n,\ \
\lambda_{t}^{(n)}=\lambda_{t,+}^{(n)}+\lambda_{t,-}^{(n)}=\sum\limits_{i=0}^nk_i(t)\lambda^{n-i}+\sum\limits_{i=n}^\infty k_i(t)\lambda^{n-i}.
\end{equation}
\eqref{13} can be broken down into
\begin{equation}\label{16}
-N_{+,x}^{(n)}+\frac{\partial M}{\partial \lambda}\lambda_{t,+}^{(n)}+[M,N_+^{(n)}]=N_{-,x}^{(n)}-\frac{\partial M}{\partial \lambda}\lambda_{t,-}^{(n)}-[M,N_-^{(n)}].
\end{equation}
 It follows that the gradations of the left-hand side of \eqref{16} are obtained as follows:
\begin{equation}\notag
\deg N_+^{(n)}=:(0,0,0)\geq 0, \ \ \deg \frac{\partial M}{\partial \lambda}\lambda_{t,+}^{(n)}=:(0,0,0)\geq 0,\ \
\deg ([M,N_+^{(n)}])=:(0,1,0;0,0,0)\geq 0,
\end{equation}
which indicate the minimum gradation of the left-hand side of \eqref{16} is zero. Similarly,    the maximum gradation of the right-hand side of \eqref{16} is also 0.
Thus, we  obtain the following equation by taking these terms which have the gradations  0:
\begin{equation}\label{17}
\begin{split}
-N_{+,x}^{(n)}+\frac{\partial M}{\partial \lambda}\lambda_{t,+}^{(n)}+[M,N_+^{(n)}]=\frac{1}{2}a_{n+1}e(0)-\frac{1}{4}c_{n+1}h(0).
\end{split}
\end{equation}
Thus, we take the modified term  $\triangle_n=-\frac{1}{4}c_{n+1}e(0)$ so that for
 $N^{(n)}=N_+^{(n)}+\triangle_n$.
By solving the nonisospectral zero curvature equation
\begin{equation}\label{18}
\frac{\partial M}{\partial u}u_t+\frac{\partial M}{\partial \lambda}\lambda_{t}^{(n)}-N_x^{(n)}+[M,N^{(n)}]=0,
\end{equation}
we obtain the  nonisospectral KdV hierarchy  as follows:
\begin{equation}\label{19}
u_{t_n}=\frac{1}{2}c_{n+1,x}=\frac{1}{2}[(\partial^2+2u+2\partial^{-1}u\partial)c_n+\frac{1}{2}k_n(t)x]_x=:\frac{1}{2}\partial Lc_n+\frac{1}{4}k_n(t),
\end{equation}
where $L=\partial^2+2u+2\partial^{-1}u\partial.$
The first two   equations in the above  hierarchy of soliton equations are
\begin{equation}\label{21}
u_{t_0}=\alpha u_x+\frac{1}{4}k_0(t),
\end{equation}
\begin{equation}\label{22}
u_{t_1}=\alpha (u_{xxx}+6uu_x)+\frac{1}{2}k_0(t)(xu_x+2u)+\frac{1}{4}k_1(t).
\end{equation}
 The equation \eqref{22} becomes the famous KdV equation when $\alpha=1,$ $k_0(t)=k_1(t)=0$.   Actually, the nonisospectral KdV hierarchy \eqref{19} can be reduced to the isospectral  KdV hierarchy by taking  $k_m(t)=0$, $m=0,1,2,\cdots.$

\subsection{Bi-Hamiltonian structures}
In this section, we discuss the Hamiltonian structure of the hierarchy \eqref{19} by means of the trace identity (see \cite{4}). Denoting the trace of the square matrices A and B by $<A,B>=tr(AB)$.
From the spectral problem \eqref{12}, we can directly calculate
\begin{equation}\notag
\begin{split}
&<N,\frac{\partial M}{\partial u}>=-2c,\ \ <N,\frac{\partial M}{\partial \lambda}>=\frac{1}{4}c.
\end{split}
\end{equation}
It follows that the   trace identity
\begin{equation}\notag
\frac{\delta}{\delta u}(<N,\frac{\partial M}{\partial \lambda}>)=\lambda^{-\gamma}\frac{\partial}{\partial\lambda}\lambda^{\gamma}
(<N,\frac{\partial M}{\partial u}>)
\end{equation}
admits that
\begin{equation}\notag
\frac{\delta}{\delta u}(\frac{1}{4}\sum\limits_{m\geq0}c_m\lambda^{-m})=\lambda^{-\gamma}\frac{\partial}{\partial\lambda}\lambda^{\gamma}
(-\sum\limits_{m\geq0}c_m\lambda^{-m}),
\end{equation}
and thus we have
\begin{equation}\notag
\frac{\delta c_m}{\delta u}=-4(\gamma-m+1)c_{m-1},\ \ m\geq1.
\end{equation}
One can find that $\gamma=-\frac{1}{2}$ via substituting $m=1$ into the above equation by means of \eqref{15*}. Therefore, we obtain
\begin{equation}\label{23}
c_m=\frac{\delta H_{m+1}}{\delta u},\ \ \ H_{m+1}=\frac{c_{m+1}}{2(2m+1)},\ \ m\geq0.
\end{equation}
Furthermore, we obtain the bi-Hamiltonian structures of the hierarchy \eqref{19} as follows:

\begin{equation}\label{24}
\begin{split}
u_{t_n}=&K_n=\frac{1}{2}\partial c_{n+1}=:Jc_{n+1}=J\frac{\delta H_{n+2}}{\delta u}\\
=&J(L\frac{\delta H_{n+1}}{\delta u}+\frac{1}{2}k_n(t)x)=:M\frac{\delta H_{n+1}}{\delta u}+\frac{1}{4}k_n(t)\\
=&J(L^{n+1}c_0+\frac{1}{2}\sum\limits_{m=0}^{n}k_m(t)L^{n-m}x)=\Phi^{n+1}Jc_0+\frac{1}{2}\sum\limits_{m=0}^{n}k_m(t)\Phi^{n-m}Jx,
\end{split}
\end{equation}
where $J$ and $M$ are two Hamiltonian operators and $\Phi$ is a hereditary symmetry operator  as follows:
\begin{equation}\label{25}
J=\frac{1}{2}\partial,\ \ M=JL=\frac{1}{2}(\partial^3+2\partial u+2u\partial),\ \ \Phi=JLJ^{-1}=\partial^2+2u_x\partial^{-1}+4u.
\end{equation}

We can conclude that the integrable hierarchy \eqref{19} is  integrable in the sense of Liouville (see \cite{200}),  and  we  obtain the Abelian algebra of symmetries
\begin{equation}\label{26}
[K_i,K_j]=K_i'(u)[K_j]-K_j'(u)[K_i]=0,\ \  i,j\geq0,
\end{equation}
and the  Abelian algebras of conserved functionals,
\begin{equation}\label{27}
\{H_{i},H_{j}\}_{J}=\int(\frac{\delta H_{i}}{\delta u})^TJ\frac{\delta H_{j}}{\delta u}dx=0,\ \ \{H_{i},H_{j}\}_{M}=\int(\frac{\delta H_{i}}{\delta u})^TJ\frac{\delta H_{j}}{\delta u}dx=0,\ \  i,j\geq0.
\end{equation}

\section{A coupled nonisospectral KdV hierarchy}

For the  Lie algebra $A_{12}$ (see eq.\eqref{3}), we consider the corresponding loop algebra
\begin{equation}\notag
\widetilde{A}_{12}=\rm span\{h_1(n), h_2(n), h_3(n), h_4(n), h_5(n), h_6(n)\}
\end{equation}
where
\begin{equation}\notag
h_i(n)=h_i\lambda^{n}, \ i=1,2,\cdots,6.
\end{equation}
Introducing the following  nonisospectral problem based on $\widetilde{A}_{11}$
\begin{equation}\label{b12}
\begin{cases}
\psi_x=U\psi,\ \ U=\frac{1}{4}h_2(1)-u_1h_2(0)+h_3(0)-u_2h_5(0)=\left[\begin{matrix}
U_1&\varepsilon U_2\\
U_2&U_1
\end{matrix}
\right], \\
\psi_t=W\psi,\ \ W=ah_1(0)+bh_2(0)+ch_3(0)+eh_4(0)+fh_5(0)+gh_6(0)=\left[\begin{matrix}
W_1&\varepsilon W_2\\
W_2&W_1
\end{matrix}
\right],\\
\lambda_t=\sum\limits_{m\geq0}k_m(t)\lambda^{-m},
\end{cases}
\end{equation}
where
 \begin{equation}\notag
 \begin{split}
 &U_1=\left(\begin{matrix}
0&-u_1+\frac{\lambda}{4}\\
1&0
\end{matrix}
\right),\ \ U_2=\left(\begin{matrix}
0&-u_2\\
0&0
\end{matrix}
\right),\ \ W_1=\left(\begin{matrix}
a&b\\
c&-a\end{matrix}
\right),\ \ W_2=\left(\begin{matrix}
e&f\\
g&-e\end{matrix}
\right)\\
 &a=\sum\limits_{m\geq0}a_m\lambda^{-m}, b=\sum\limits_{m\geq0}b_m\lambda^{-m}, c=\sum\limits_{m\geq0}c_m\lambda^{-m}, e=\sum\limits_{m\geq0}e_m\lambda^{-m}, f=\sum\limits_{m\geq0}f_m\lambda^{-m},  g=\sum\limits_{m\geq0}g_m\lambda^{-m}.
  \end{split}
 \end{equation}
The   stationary zero-curvature equation
\begin{equation}\label{b13}
W_x=\frac{\partial U}{\partial \lambda}\lambda_t+[U,W],
\end{equation}
gives rise to the  recursion equations:
\begin{equation}\label{b14}
\begin{cases}
a_{mx}=-u_1c_{m}+\frac{1}{4}c_{m+1}-b_m-\varepsilon u_2g_{m},\\
b_{mx}=\frac{1}{4}k_m(t)-\frac{1}{2}a_{m+1}+2u_1a_{m}+2\varepsilon u_2e_{m},\\
c_{mx}=2a_m,\\
e_{mx}=-u_1g_{m}+\frac{1}{4}g_{m+1}-f_m-u_2c_{m},\\
f_{mx}=-\frac{1}{2}e_{m+1}+2u_1e_{m}+2u_2a_{m},\\
g_{mx}=2e_m,
\end{cases}
\end{equation}
which have the  equivalent forms
\begin{equation}\label{b15}
\begin{cases}
a_{m}=\frac{1}{2}c_{mx},\\
b_m=-\frac{1}{4}c_{m+1}+\partial^{-1}(u_1c_{mx}+\varepsilon u_2g_{mx})+\frac{1}{4}k_m(t)x,\\
c_{m+1}=(\partial^2+2u_1+2\partial^{-1}u_1\partial)c_m+(2\varepsilon u_2+2\varepsilon \partial^{-1}u_2\partial)g_m+\frac{1}{2}k_m(t)x,\\
e_{m}=\frac{1}{2}g_{mx},\\
f_m=-\frac{1}{4}g_{m+1}+\partial^{-1}(u_1g_{mx}+u_2c_{mx}),\\
g_{m+1}=(2u_2+2\partial^{-1}u_2\partial)c_m+(\partial^2+2u_1+2\partial^{-1}u_1\partial)g_m.
\end{cases}
\end{equation}
To the  recursion equations \eqref{b15}, we take the initial values $ a_0=e_0=0,$
it follows that one has
\begin{equation}\label{b15*}
\begin{split}
& a_0=e_0=0,\ \ c_0=\alpha_1,\ \ g_0=\alpha_2,\ \ c_1=2\alpha_1 u_1+2\varepsilon\alpha_2 u_2+\frac{1}{2}k_0(t)x,\ \ g_1=2\alpha_2 u_1+2\alpha_1 u_2,\\
&b_0=-\frac{1}{2}\alpha_1 u_1-\frac{\varepsilon}{2}\alpha_2 u_2+\frac{1}{8}k_0(t)x,\ \
f_0=-\frac{1}{2}\alpha_2 u_1-\frac{1}{2}\alpha_1 u_2,\\
&a_1=\alpha_1 u_{1x}+\varepsilon\alpha_2 u_{2x}+\frac{1}{4}k_0(t),\ \
e_1=\alpha_2 u_{1x}+\alpha_1 u_{2x},\\
 &c_2=2\alpha_1 u_{1xx}+2\varepsilon\alpha_2 u_{2xx}
+6\alpha_1 u_1^2+6\varepsilon\alpha_1 u_2^2+12\varepsilon\alpha_2 u_1u_2+k_0(t)xu_1+k_0(t)\partial^{-1}u_1+\frac{1}{2}k_1(t)x,\\
&g_2=2\alpha_2 u_{1xx}+2\alpha_1 u_{2xx}+6\alpha_2 u_1^2+6\varepsilon\alpha_2 u_2^2+12\alpha_1 u_1u_2+k_0(t)xu_2+k_0(t)\partial^{-1}u_2,\\
&\cdots
\end{split}
\end{equation}

Denoting that
\begin{equation}\notag
W^{(n)}=W\lambda^n,\ \
W_+^{(n)}=\sum\limits_{i=0}^n(a_ih_1(n-i)+b_ih_2(n-i)+c_ih_3(n-i)+e_ih_4(n-i)+f_ih_5(n-i)+g_ih_6(n-i)),
\end{equation}
\eqref{b13} can be broken down into
\begin{equation}\label{b16}
-W_{+,x}^{(n)}+\frac{\partial U}{\partial \lambda}\lambda_{t,+}^{(n)}+[U,W_+^{(n)}]=W_{-,x}^{(n)}-\frac{\partial U}{\partial \lambda}\lambda_{t,-}^{(n)}-[U,W_-^{(n)}].
\end{equation}
The minimum gradation of the left-hand side and the maximum gradation of the right-hand side of \eqref{b16} are both zero.
It follows that one has:
\begin{equation}\label{b17}
\begin{split}
-W_{+,x}^{(n)}+\frac{\partial U}{\partial \lambda}\lambda_{t,+}^{(n)}+[U,W_+^{(n)}]=\frac{1}{2}a_{n+1}e(0)-\frac{1}{4}c_{n+1}h(0).
\end{split}
\end{equation}
Thus, we take the modified term  $\triangle_n=-\frac{1}{4}c_{n+1}h_2(0)-\frac{1}{4}g_{n+1}h_5(0)$ so that for
 $V^{(n)}=W_+^{(n)}+\triangle_n$.
The nonisospectral zero curvature equation
\begin{equation}\label{b18}
\frac{\partial U}{\partial u}u_t+\frac{\partial U}{\partial \lambda}\lambda_{t}^{(n)}-V_x^{(n)}+[U,V^{(n)}]=0,
\end{equation}
leads to the  coupled nonisospectral KdV hierarchy  as follows:
\begin{equation}\label{b19}
u_{t_n}=\overline{K}_n=\left(\begin{matrix}
\frac{1}{2}c_{n+1,x}\\
\frac{1}{2}g_{n+1,x}
\end{matrix}
\right)=J\left(\begin{matrix}
c_{n+1}\\
g_{n+1}
\end{matrix}
\right)
=J[\overline{L}\left(\begin{matrix}
c_{n}\\
g_{n}
\end{matrix}
\right)+\frac{1}{2}k_n(t)\left(\begin{matrix}
x\\
0
\end{matrix}
\right)],
\end{equation}
where $J$ is given by \eqref{25} and
\begin{equation}\notag
 \overline{L}=\left(\begin{matrix}
\partial^2+4u_1-2\partial^{-1}u_{1x}&\varepsilon(4u_2-2\partial^{-1}u_{2x})\\
4u_2-2\partial^{-1}u_{2x}&\partial^2+4u_1-2\partial^{-1}u_{1x}
\end{matrix}
\right)=\left(\begin{matrix}
\overline{L}_{11}&\varepsilon\overline{L}_{21}\\
\overline{L}_{21}&\overline{L}_{11}
\end{matrix}
\right),
\end{equation}
which is determined by the  recursion equations \eqref{b15}.
The first two   examples in the above  hierarchy of soliton equations are
\begin{equation}\label{b21}
u_{t_0}=\frac{1}{2}\left(\begin{matrix}
c_{1,x}\\
g_{1,x}
\end{matrix}
\right)=\alpha_1\left(\begin{matrix}
u_{1x}\\
u_{2x}
\end{matrix}
\right)+\alpha_2\left(\begin{matrix}
\varepsilon u_{2x}\\
u_{1x}
\end{matrix}
\right)
+\frac{1}{4}k_0(t)\left(\begin{matrix}
1\\
0
\end{matrix}
\right),
\end{equation}
\begin{equation}\label{b22}
\begin{split}
u_{t_1}=&\alpha_1\left(\begin{matrix}
u_{1xxx}+6u_1u_{1x}+6\varepsilon u_2u_{2x}\\
u_{2xxx}+6u_1u_{2x}+6u_2u_{1x}
\end{matrix}
\right)+\alpha_2\left(\begin{matrix}
\varepsilon u_{2xxx}+6\varepsilon u_2u_{1x}+6\varepsilon u_1u_{2x}\\
u_{1xxx}+6u_1u_{1x}+6\varepsilon u_2u_{2x}
\end{matrix}
\right)\\
&+\frac{1}{2}k_0(t)\left(\begin{matrix}
2u_1+xu_{1x}\\
2u_2+xu_{2x}
\end{matrix}
\right)+\frac{1}{4}k_1(t)\left(\begin{matrix}
1\\
0
\end{matrix}
\right),
\end{split}
\end{equation}
which is a coupled nonisospectral KdV equation.\\
Taking $\alpha_1=1,$ $\alpha_2=0,$ $k_0(t)=k_1(t)=0$,
and \eqref{b22} becomes the Frobenius KdV equation
\begin{equation}\label{b23}
\begin{cases}
u_{1t}=u_{1xxx}+6u_1u_{1x}+6\varepsilon u_2u_{2x},\\
u_{2t}=u_{2xxx}+6u_1u_{2x}+6u_2u_{1x},
\end{cases}
\end{equation}
which  is derived by  the Frobenius-Virasoro algebra (see \cite{201,202}).
\eqref{b23} can be reduced to the coupled KdV equation (see \cite{203,204}) and the complexification of the KdV equation by taking $\varepsilon=0$ and $\varepsilon=-1$ respectively.
From \eqref{b12}, the integrable system \eqref{b23}  has the following Lax pairs:
\begin{equation}
\begin{split}
&\psi_x=U\psi=\left[\begin{matrix}
U_1&\varepsilon U_2\\
U_2&U_1
\end{matrix}
\right]\psi,\ \ U_1=\left(\begin{matrix}
0&-u_1+\frac{\lambda}{4}\\
1&0
\end{matrix}
\right),\ \ U_2=\left(\begin{matrix}
0&-u_2\\
0&0
\end{matrix}
\right),\\
&\psi_t=V\psi=\left[\begin{matrix}
V_1&\varepsilon V_2\\
V_2&V_1
\end{matrix}
\right]\psi, \ \ V_1=\left(\begin{matrix}
u_{1x}&-u_{1xx}-2u_1^2-2\varepsilon u_2^2-\frac{1}{2}\lambda u_1+\frac{1}{4}\lambda^2\\
2u_1+\lambda&-u_{1x}
\end{matrix}
\right),\\
&\ \ \ \ \ \ \ \ \ \ \ \ \ \ \ \ \ \ \ \ \ \ \ \ \ \ \ \ \ \ \ \ \ \ \ V_2=\left(\begin{matrix}
u_{2x}&-u_{2xx}-4u_1u_2-\frac{1}{2}\lambda u_2\\
2u_2&-u_{2x}
\end{matrix}
\right).
\end{split}
\end{equation}
\subsection{Bi-Hamiltonian structures}

In order to establish the bi-Hamiltonian structures of  the coupled nonisospectral KdV hierarchy \eqref{b19},  we use the following two sets of component-trace identity (see \cite{205})
\begin{equation}\notag
\begin{cases}
\frac{\delta}{\delta u}(<W_1,\frac{\partial U_2}{\partial \lambda}>+<W_2,\frac{\partial U_1}{\partial \lambda}>)=\lambda^{-\gamma}\frac{\partial}{\partial\lambda}\lambda^{\gamma}
(<W_1,\frac{\partial U_2}{\partial u}>+<W_2,\frac{\partial U_1}{\partial u}>),\\
\frac{\delta}{\delta u}(<W_1,\frac{\partial U_1}{\partial \lambda}>+\varepsilon<W_2,\frac{\partial U_2}{\partial \lambda}>)=\lambda^{-\gamma}\frac{\partial}{\partial\lambda}\lambda^{\gamma}
(<W_1,\frac{\partial U_1}{\partial u}>+\varepsilon<W_2,\frac{\partial U_2}{\partial u}>).
\end{cases}
\end{equation}
Substituting
\begin{equation}\notag
\begin{split}
&<W_1,\frac{\partial U_2}{\partial u}>+<W_2,\frac{\partial U_1}{\partial u}>=\left(\begin{matrix}
-g\\
-c\end{matrix}
\right),\ \
<W_1,\frac{\partial U_1}{\partial u}>+\varepsilon<W_2,\frac{\partial U_2}{\partial u}>=\left(\begin{matrix}
-c\\
-\varepsilon g\end{matrix}
\right),\\
 &<W_1,\frac{\partial U_2}{\partial \lambda}>+<W_2,\frac{\partial U_1}{\partial \lambda}>=\frac{1}{4}g,\ \ <W_1,\frac{\partial U_1}{\partial \lambda}>+\varepsilon<W_2,\frac{\partial U_2}{\partial \lambda}>=\frac{1}{4}c,\\
\end{split}
\end{equation}
 into the above  component-trace identity
 and comparing powers of $\lambda$, we obtain
\begin{equation}\label{b24}
\frac{\delta c_m}{\delta u}=-4(\gamma-m+1)\left(\begin{matrix}
c_{m-1}\\
\varepsilon g_{m-1}\end{matrix}
\right),\ \ \frac{\delta g_m}{\delta u}=-4(\gamma-m+1)\left(\begin{matrix}
g_{m-1}\\
c_{m-1}\end{matrix}
\right),\ \ m\geq1.
\end{equation}
where $W_1$,  $W_2$, $U_1$,  $U_2$ is given by \eqref{b12}.
One can find that $\gamma=-\frac{1}{2}$ via substituting $m=1$ into \eqref{b24} by means of \eqref{15*}. Therefore, we obtain
\begin{equation}\label{b25}
\begin{split}
&\left(\begin{matrix}
c_m\\
\varepsilon g_{m}\end{matrix}
\right)=\frac{\delta H_{1,m+1}}{\delta u},\ \
\left(\begin{matrix}
g_m\\
c_{m}\end{matrix}
\right)=\frac{\delta H_{2,m+1}}{\delta u},\ \ m\geq0.
\end{split}
\end{equation}
where
$ H_{1,m+1}=\frac{c_{m+1}}{2(2m+1)}$ and  $ H_{2,m+1}=\frac{g_{m+1}}{2(2m+1)}$  are the Hamiltonian functionals.
From \eqref{b25}, one can find  that the Hamiltonian structure of the hierarchy \eqref{b19} consists of  two component and   let us discuss the two components separately.

For the fist component,  we have:
\begin{equation}\label{b27}
\begin{split}
u_{t_n}=&\left(\begin{matrix}
\frac{1}{2}c_{n+1,x}\\
\frac{1}{2}g_{n+1,x}
\end{matrix}
\right)=\frac{1}{2}\left(\begin{matrix}
\partial&0\\
0&\frac{1}{\varepsilon}\partial
\end{matrix}
\right)\left(\begin{matrix}
c_{n+1}\\
\varepsilon g_{n+1}
\end{matrix}
\right)=:J_1\left(\begin{matrix}
c_{n+1}\\
\varepsilon g_{n+1}
\end{matrix}
\right)=J_1\frac{\delta H_{1,n+2}}{\delta u}\\
=&J_1(\overline{L}^T\frac{\delta H_{1,n+1}}{\delta u}+k_n(t)\left(\begin{matrix}
\frac{x}{2}\\
0\end{matrix}
\right))=:\overline{M}_1\frac{\delta H_{1,n+1}}{\delta u}+k_n(t)J_1\left(\begin{matrix}
\frac{x}{2}\\
0\end{matrix}
\right)\\
=&J_1[(\overline{L}^T)^{n+1}\frac{\delta H_{1,1}}{\delta u}+\sum\limits_{m=0}^{n}k_{m-1}(t)(\overline{L}^T)^{n-m}\left(\begin{matrix}
\frac{x}{2}\\
0\end{matrix}
\right)]\\
=&\overline{\Phi}_1^{n+1}J_1\frac{\delta H_{1,1}}{\delta u}+\sum\limits_{m=0}^{n}k_{m-1}(t)\overline{\Phi}_1^{n-m}J_1\left(\begin{matrix}
\frac{x}{2}\\
0\end{matrix}
\right),\ \ \ \ n\geq1,
\end{split}
\end{equation}
where $\overline{L}^T$ is the device matrix of matrix $\overline{L}$ in \eqref{b19}, $J_1$ and $\overline{M}_1$ are two Hamiltonian operators and $\overline{\Phi}_1$ is a hereditary symmetry operator  as follows:
\begin{equation}\label{b28}
\begin{split}
&J_1=\frac{1}{2}\left(\begin{matrix}
\partial&0\\
0&\frac{1}{\varepsilon}\partial
\end{matrix}
\right),\ \ \overline{M}_1=J_1\overline{L}^T=\frac{1}{2}\left(\begin{matrix}
\partial^3+2\partial u_1+2u_1\partial&2\partial u_2+2u_2\partial\\
2\partial u_2+2u_2\partial&\frac{1}{\varepsilon}(\partial^3+2\partial u_1+2u_1\partial)
\end{matrix}
\right),\\
&\overline{\Phi}_1=J_1\overline{L}^TJ_1^{-1}=\left(\begin{matrix}
\partial^2+4u_1+2u_{1x}\partial^{-1}&\varepsilon(4u_2+2u_{2x}\partial^{-1})\\
4u_2+2u_{2x}\partial^{-1}&\partial^2+4u_1+2u_{1x}\partial^{-1}
\end{matrix}
\right).
\end{split}
\end{equation}
For the second component,  one has:
\begin{equation}\label{b29}
\begin{split}
u_{t_n}=&\left(\begin{matrix}
\frac{1}{2}c_{n+1,x}\\
\frac{1}{2}g_{n+1,x}
\end{matrix}
\right)=\frac{1}{2}\left(\begin{matrix}
0&\partial\\
\partial&0
\end{matrix}
\right)\left(\begin{matrix}
g_{n+1}\\
c_{n+1}
\end{matrix}
\right)=:J_2\left(\begin{matrix}
g_{n+1}\\
c_{n+1}
\end{matrix}
\right)=J_2\frac{\delta H_{2,n+2}}{\delta u}\\
=&J_2(\overline{L}^T\frac{\delta H_{2,n+1}}{\delta u}+k_n(t)\left(\begin{matrix}
0\\
\frac{x}{2}\end{matrix}
\right))=:\overline{M}_2\frac{\delta H_{2,n+1}}{\delta u}+k_n(t)J_2\left(\begin{matrix}
0\\
\frac{x}{2}\end{matrix}
\right)\\
=&J_2[(\overline{L}^T)^{n+1}\frac{\delta H_{2,1}}{\delta u}+\sum\limits_{m=0}^{n}k_{m-1}(t)(\overline{L}^T)^{n-m}\left(\begin{matrix}
0\\
\frac{x}{2}\end{matrix}
\right)]\\
=&\overline{\Phi}_2^{n+1}J_2\frac{\delta H_{2,1}}{\delta u}+\sum\limits_{m=0}^{n}k_{m-1}(t)\overline{\Phi}_2^{n-m}J_2\left(\begin{matrix}
0\\
\frac{x}{2}\end{matrix}
\right),\ \ \ \ n\geq1,
\end{split}
\end{equation}
where  $J_2$ and $\overline{M}_2$ are two Hamiltonian operators and $\overline{\Phi}_2$ is a hereditary symmetry operator  as follows:
\begin{equation}\label{b30}
\begin{split}
&J_2=\frac{1}{2}\left(\begin{matrix}
0&\partial\\
\partial&0
\end{matrix}
\right),\ \ \overline{M}_2=J_2\overline{L}^T=\frac{1}{2}\left(\begin{matrix}
\varepsilon(2\partial u_2+2u_2\partial)&\partial^3+2\partial u_1+2u_1\partial\\
\partial^3+2\partial u_1+2u_1\partial&2\partial u_2+2u_2\partial
\end{matrix}
\right),\\
&\overline{\Phi}_2=J_2\overline{L}^TJ_2^{-1}=\left(\begin{matrix}
\partial^2+4u_1+2u_{1x}\partial^{-1}&\varepsilon(4u_2+2u_{2x}\partial^{-1})\\
4u_2+2u_{2x}\partial^{-1}&\partial^2+4u_1+2u_{1x}\partial^{-1}
\end{matrix}
\right)=\overline{\Phi}_1=:\overline{\Phi}.
\end{split}
\end{equation}
After comparing and analyzing the results of these two components, we obtain the bi-Hamiltonian structures of  the coupled nonisospectral KdV hierarchy \eqref{b19} as follows:
\begin{equation}\label{b31}
\begin{cases}
u_{t_n}=J_1\frac{\delta H_{1,n+2}}{\delta u}=\overline{M}_1\frac{\delta H_{1,n+2}}{\delta u}+k_n(t)J_1\left(\begin{matrix}
\frac{x}{2}\\
0\end{matrix}
\right)=\overline{\Phi}^nJ_1\left(\begin{matrix}
c_1\\
\varepsilon g_1\end{matrix}
\right)+\frac{1}{2}\sum\limits_{m=1}^{n}k_{m}(t)\overline{\Phi}^{n-m}J_1\left(\begin{matrix}
\frac{x}{2}\\
0\end{matrix}
\right),
\\
u_{t_n}=J_2\frac{\delta H_{2,n+2}}{\delta u}=\overline{M}_2\frac{\delta H_{2,n+2}}{\delta u}+k_n(t)J_2\left(\begin{matrix}
0\\
\frac{x}{2}\end{matrix}
\right)=\overline{\Phi}^nJ_2\left(\begin{matrix}
g_1\\
c_1\end{matrix}
\right)+\frac{1}{2}\sum\limits_{m=1}^{n}k_{m}(t)\overline{\Phi}^{n-m}J_2\left(\begin{matrix}
0\\
\frac{x}{2}\end{matrix}
\right),
\end{cases}
\end{equation}
with the recursion operator
\begin{equation}\notag
\overline{\Phi}=\left(\begin{matrix}
\partial^2+4u_1+2u_{1x}\partial^{-1}&\varepsilon(4u_2+2u_{2x}\partial^{-1})\\
4u_2+2u_{2x}\partial^{-1}&\partial^2+4u_1+2u_{1x}\partial^{-1}
\end{matrix}
\right).
\end{equation}
After a simple calculation, one can find that the two equations of \eqref{b31} can be combined into one equation as follows:
\begin{equation}\label{b32}
u_{t_n}=\overline{K}_n=\alpha_1\overline{\Phi}^n\left(\begin{matrix}
u_{1x}\\
u_{2x}\end{matrix}
\right)+\alpha_2\overline{\Phi}^n\left(\begin{matrix}
\varepsilon u_{2x}\\
u_{1x}\end{matrix}
\right)+\frac{1}{2}\sum\limits_{m=0}^{n}k_{m}(t)\overline{\Phi}^{n-m}\left(\begin{matrix}
\frac{1}{2}\\
0\end{matrix}
\right).
\end{equation}
which is  equivalent to the coupled nonisospectral KdV hierarchy \eqref{b19}.
 Selecting $\alpha_1=1,$ $\alpha_2=0,$ and then \eqref{b32} becomes
\begin{equation}\label{b33}
u_{t_n}=\overline{K}_n=\overline{\Phi}^n\left(\begin{matrix}
u_{1x}\\
u_{2x}\end{matrix}
\right)+\frac{1}{2}\sum\limits_{m=0}^{n}k_{m}(t)\overline{\Phi}^{n-m}\left(\begin{matrix}
\frac{1}{2}\\
0\end{matrix}
\right),
\end{equation}
where $\overline{\Phi}$ is given by \eqref{b31}.

Obviously, both $\overline{M}_1$  and $\overline{M}_2$ of \eqref{b31} are antisymmetric operators, that is $\overline{M}_1^\ast=-\overline{M}_1$, $\overline{M}_2^\ast=-\overline{M}_2$, and thus we have
\begin{equation}\notag
\overline{M}_1=\overline{\Phi}J_1=J_1\overline{\Phi}^\ast,\ \ \overline{M}_2=\overline{\Phi}J_2=J_2\overline{\Phi}^\ast,
\end{equation}
where $\overline{\Phi}^\ast$  is the complex conjugate of $\overline{\Phi}$.
We can conclude that the integrable hierarchy \eqref{b19} is  integrable in the sense of Liouville,  and then we  obtain the Abelian algebra of symmetries
\begin{equation}\label{b34}
[\overline{K}_i,\overline{K}_j]=\overline{K}_i'(u)[\overline{K}_j]-\overline{K}_j'(u)[\overline{K}_i]=0,\ \  i,j\geq0,
\end{equation}
and the component Abelian algebras of conserved functionals,
\begin{equation}\label{b35}
\begin{split}
\{H_{1,i},H_{1,j}\}_{J_1}=\int(\frac{\delta H_{1,i}}{\delta u})^TJ_1\frac{\delta H_{1,j}}{\delta u}dx=0,\ \ \{H_{1,i},H_{1,j}\}_{\overline{M}_1}=\int(\frac{\delta H_{1,i}}{\delta u})^TJ_1\frac{\delta H_{1,j}}{\delta u}dx=0,\\
\{H_{2,i},H_{2,j}\}_{J_2}=\int(\frac{\delta H_{2,i}}{\delta u})^TJ_2\frac{\delta H_{2,j}}{\delta u}dx=0,\ \ \{H_{2,i},H_{2,j}\}_{\overline{M}_2}=\int(\frac{\delta H_{2,i}}{\delta u})^TJ_2\frac{\delta H_{2,j}}{\delta u}dx=0.
\end{split}
\end{equation}

\section{Symmetries and their Lie algebra of the coupled nonisospectral KdV hierarchy}

In this section, we will discuss the $K$ symmetries, $\tau$ symmetries and their Lie algebra of the coupled nonisospectral KdV hierarchy \eqref{b33}.
Let us first recall some basic notions and properties related to symmetries. Assume $G(u)=U(u,u_x,\ldots)$ is a function, then the Gateaux derivative can be  defined as:
$
G'(u)[\sigma]=\lim_{\eta\rightarrow0}\frac{d}{d\eta}G(u+\eta\sigma).
$
$\sigma_t=K_n'[\sigma]$ is a linearised equation associated with a given  soliton hierarchy $u_t=K_n$, and then we call $\sigma$  a symmetry    of the  given soliton hierarchy.
 In integrable hierarchy \eqref{b33}, one has:
\begin{equation}\label{b36}
\begin{split}
&\overline{\Phi}=\left(\begin{matrix}
\partial^2+4u_1+2u_{1x}\partial^{-1}&\varepsilon(4u_2+2u_{2x}\partial^{-1})\\
4u_2+2u_{2x}\partial^{-1}&\partial^2+4u_1+2u_{1x}\partial^{-1}
\end{matrix}
\right)=:\phi,\  \ K_0=\left(\begin{matrix}
u_{1x}\\
u_{2x}\end{matrix}
\right),\ \ \sigma_0=\left(\begin{matrix}
\frac{1}{2}\\
\frac{1}{2}\end{matrix}
\right),\\
&K_1=\phi K_0=\left(\begin{matrix}
u_{1xxx}+6u_1u_{1x}+6\varepsilon u_2u_{2x}\\
u_{2xxx}+6u_1u_{2x}+6u_2u_{1x}
\end{matrix}
\right).
\end{split}
\end{equation}
Here we use $\phi$ to represent $\overline{\Phi}$ for convenience.

\begin{lemma}
$\phi$ is  a hereditary symmetry of the hierarchy \eqref{b33}, that is,
\begin{equation}\label{b37}
\phi'[\phi f]g-\phi'[\phi g]f=\phi\{\phi'[f]g-\phi'[g]f\}.
\end{equation}
\end{lemma}
\begin{proof}
This theorem can be directly calculated via the definition of Gateaux derivative.\\
For any $f=\left(\begin{matrix}
f_1\\
f_2\end{matrix}
\right)$, $g=\left(\begin{matrix}
g_1\\
g_2\end{matrix}
\right)$, one has
\begin{equation}\notag
\phi'[f]=\left(\begin{matrix}
4f_1+2f_{1x}\partial^{-1}&\varepsilon(4f_2+2f_{2x}\partial^{-1})\\
4f_2+2f_{2x}\partial^{-1}&4f_1+2f_{1x}\partial^{-1}
\end{matrix}
\right).
\end{equation}
It follows that we obtain
\begin{equation}\notag
\phi'[f]g-\phi'[g]f=\left(\begin{matrix}
2f_{1x}\partial^{-1}g_1-2g_{1x}\partial^{-1}f_1+2\varepsilon f_{2x}\partial^{-1}g_2-2\varepsilon g_{2x}\partial^{-1}f_2\\
2f_{2x}\partial^{-1}g_1-2g_{2x}\partial^{-1}f_1+2f_{2x}\partial^{-1}g_2-2g_{2x}\partial^{-1}f_2
\end{matrix}
\right).
\end{equation}
Since
\begin{equation}\notag
\phi f=\left(\begin{matrix}
f_{1xx}+4u_1f_1+2u_{1x}\partial^{-1}f_1+\varepsilon(4u_2f_2+2u_{2x}\partial^{-1}f_2)\\
f_{2xx}+4u_1f_2+2u_{1x}\partial^{-1}f_2+4u_2f_1+2u_{2x}\partial^{-1}f_1
\end{matrix}
\right)=:\left(\begin{matrix}
s_1\\
s_2
\end{matrix}
\right),
\end{equation}
\begin{equation}\notag
\phi g=\left(\begin{matrix}
g_{1xx}+4u_1g_1+2u_{1x}\partial^{-1}g_1+\varepsilon(4u_2g_2+2u_{2x}\partial^{-1}g_2)\\
g_{2xx}+4u_1g_2+2u_{1x}\partial^{-1}g_2+4u_2g_1+2u_{2x}\partial^{-1}g_1
\end{matrix}
\right)=:\left(\begin{matrix}
s_3\\
s_4
\end{matrix}
\right),
\end{equation}
then
\begin{equation}\notag
\phi'[\phi f]g-\phi'[\phi g]f=2\left(\begin{matrix}
s_{1x}\partial^{-1}g_1-s_{3x}\partial^{-1}f_1+2s_1g_1-2s_3f_1
+\varepsilon(s_{2x}\partial^{-1}g_2-s_{4x}\partial^{-1}f_2+2s_2g_2-2s_4f_2)\\
s_{2x}\partial^{-1}g_1-s_{4x}\partial^{-1}f_1+2s_2g_1-2s_4f_1
+s_{1x}\partial^{-1}g_1-s_{3x}\partial^{-1}f_4+2s_1g_2-2s_3f_2
\end{matrix}
\right).
\end{equation}
After a direct calculation, one can find that \eqref{b37} holds.
\end{proof}

\begin{lemma}
$\phi$ is  a strong symmetry of the hierarchy \eqref{b33}, that is,
\begin{equation}\label{b38}
\phi'[K_m]=[K_m',\phi],\  \ m\geq0.
\end{equation}
\end{lemma}
\begin{proof}
A direct calculation gives rise to
\begin{equation}\notag
\phi'[K_0]=\left(\begin{matrix}
2u_{1xx}\partial^{-1}+4u_{1x}&\varepsilon(2u_{2xx}\partial^{-1}+4u_{2x})\\
2u_{2xx}\partial^{-1}+4u_{2x}&2u_{1xx}\partial^{-1}+4u_{1x}
\end{matrix}
\right).
\end{equation}
For a   test function $\sigma=\left(\begin{matrix}
\sigma_1,\sigma_2\end{matrix}
\right)^T$,  we have
\begin{equation}\notag
\begin{split}
(K_0)'[\sigma]=\frac{d}{d\eta}\mid_{\eta=0}\left(\begin{matrix}
(u_1+\eta\sigma_1)_x\\
(u_2+\eta\sigma_2)_x\end{matrix}
\right)=\partial\left(\begin{matrix}
\sigma_1\\
\sigma_2\end{matrix}
\right)\ \Longrightarrow \ \
(K_0)'=\partial.
\end{split}
\end{equation}
Hence, we obtain
\begin{equation}\notag
K_0'\phi=\left(\begin{matrix}
\partial^3+6u_{1x}+4u_1\partial+2u_{1xx}\partial^{-1}&\varepsilon(6u_{2x}+4u_2\partial+2u_{2xx}\partial^{-1})\\
6u_{2x}+4u_2\partial+2u_{2xx}\partial^{-1}&\partial^3+6u_{1x}+4u_1\partial+2u_{1xx}\partial^{-1}
\end{matrix}
\right),
\end{equation}
\begin{equation}\notag
\phi K_0'=\left(\begin{matrix}
\partial^3+2u_{1x}+4u_1\partial&\varepsilon(2u_{2x}+4u_2\partial)\\
2u_{2x}+4u_2\partial&\partial^3+2u_{1x}+4u_1\partial
\end{matrix}
\right).
\end{equation}
It follows that we have $\phi'[K_0]=[K_0',\phi]=K_0'\Phi-\Phi K_0'$.
Therefore, we  verify that  \eqref{b38} holds since the hereditary symmetry property of $\phi$.
\end{proof}

\begin{theorem}
\begin{equation}\label{b39}
[K_m,K_n]=K_m'[K_n]-K_n'[K_m]=0, \ \ m,n=0,1,2,\cdots
\end{equation}
where $K_m=\phi^m\left(\begin{matrix}
u_{1x}\\
u_{2x}\end{matrix}
\right)$, \ \ $K_n=\phi^n\left(\begin{matrix}
u_{1x}\\
u_{2x}\end{matrix}
\right)$.
\end{theorem}
\begin{proof}
Based on lemma 1, lemma 2,  the formula
\begin{equation}\label{b40}
(\phi G)'[\sigma]=\phi'[\sigma]G+\Phi G'[\sigma],
\end{equation}
and the operator $\phi$, we can verify that \eqref{b39} holds (see \cite{208}).
\end{proof}

\begin{lemma}
\begin{equation}\label{b41}
[K_m,\sigma_0]=K_m'\left[\begin{matrix}
\frac{1}{2}\\
\frac{1}{2}\end{matrix}
\right]=(2m+1)HK_{m-1}, \ \ m=1,2,\cdots
\end{equation}
where $H=\left(\begin{matrix}
1&\varepsilon\\
1&1\end{matrix}
\right)$.
\end{lemma}
\begin{proof}
Based on \eqref{b36}, one has
\begin{equation}\notag
[K_1,\sigma_0]=\left(\begin{matrix}
u_{1xxx}+6u_1u_{1x}+6\varepsilon u_2u_{2x}\\
u_{2xxx}+6u_1u_{2x}+6u_2u_{1x}
\end{matrix}
\right)'\left[\begin{matrix}
\frac{1}{2}\\
\frac{1}{2}\end{matrix}
\right]=3\left(\begin{matrix}
u_{1x}+\varepsilon u_{2x}\\
u_{1x}+u_{2x}\end{matrix}
\right)=3HK_0.
\end{equation}
Assume \eqref{b41} holds for $m=l-1$, that is
\begin{equation}\notag
[K_{l-1},\sigma_0]=(2k-1)HK_{l-2}, \ \ l\geq3.
\end{equation}
Since
\begin{equation}\label{b42}
\phi'\sigma_0=\phi'\left[\begin{matrix}
\frac{1}{2}\\
\frac{1}{2}\end{matrix}
\right]=2H.
\end{equation}
It follows that  one has
\begin{equation}\notag
\begin{split}
[K_{l},\sigma_0]&=(\phi\phi^{l-1}\left(\begin{matrix}
u_{1x}\\
u_{2x}\end{matrix}
\right))'\left[\begin{matrix}
\frac{1}{2}\\
\frac{1}{2}\end{matrix}
\right]=\phi'\left[\begin{matrix}
\frac{1}{2}\\
\frac{1}{2}\end{matrix}
\right]\phi^{l-1}\left(\begin{matrix}
u_{1x}\\
u_{2x}\end{matrix}
\right)+\phi(\phi^{l-1}\left(\begin{matrix}
u_{1x}\\
u_{2x}\end{matrix}
\right))'\left[\begin{matrix}
\frac{1}{2}\\
\frac{1}{2}\end{matrix}
\right]\\
&=2H\phi^{l-1}\left(\begin{matrix}
u_{1x}\\
u_{2x}\end{matrix}
\right)+\phi[K_{l-1},\sigma_0]=2HK_{l-1}+\phi(2l-1)HK_{l-2}=(2l+1)HK_{l-1}.
\end{split}
\end{equation}
Therefore, \eqref{b41} holds by induction.
\end{proof}

\begin{lemma}
\begin{equation}\label{b43}
[K_m,\phi^n\sigma_0]=(2m+1)HK_{m+n-1}, \ \ m=1,2,\cdots,\ \ n=0,1,2,\cdots.
\end{equation}
\end{lemma}
\begin{proof}
From Lemma 3, one can find  that \eqref{b43} holds when n=1.
Assume \eqref{b43} holds for $n=l-1$, that is
\begin{equation}\notag
[K_m,\phi^{l-1}\sigma_0]=(2m+1)HK_{m+l-2}, \ \ l\geq3.
\end{equation}
Then, we have
\begin{equation}\notag
\begin{split}
[K_{m},\phi^{l}\sigma_0]&=K_{m}'[\phi^{l}\sigma_0]-(\phi^{l}\sigma_0)'[K_{m}]\\
&=K_{m}'[\phi^{l}\sigma_0]-\phi'[K_{m}]\phi^{l-1}\sigma_0-\phi(\phi^{l-1}\sigma_0)'[K_{m}]\\
&=K_{m}'[\phi^{l}\sigma_0]-(K_{m}'\phi-\phi K_{m}')\phi^{l-1}\sigma_0-\phi(\phi^{l-1}\sigma_0)'[K_{m}]\\
&=\phi[K_m,\phi^{l-1}\sigma_0]=(2m+1)HK_{m+l-1}.
\end{split}
\end{equation}
Hence, we conclude that \eqref{b43} holds by induction.
\end{proof}

\begin{lemma}
\begin{equation}\label{b44}
[\phi^m\sigma_0,\sigma_0]=2m\phi^{m-1}H\sigma_0, \ \ m=1,2,\cdots.
\end{equation}
\end{lemma}
\begin{proof}
When $m=1$, one has
\begin{equation}\notag
[\phi\sigma_0,\sigma_0]=(\phi\sigma_0)'\sigma_0=\left(\begin{matrix}
xu_{1x}+2u_1+\varepsilon xu_{2x}+2\varepsilon u_2\\
xu_{2x}+2u_2+xu_{1x}+2u_1
\end{matrix}
\right)'\left[\begin{matrix}
\frac{1}{2}\\
\frac{1}{2}\end{matrix}
\right]=\left(\begin{matrix}
1+\varepsilon\\
2\end{matrix}
\right)=2H\sigma_0.
\end{equation}
Assume
\begin{equation}\notag
[\phi^{l-1}\sigma_0,\sigma_0]=2(l-1)\phi^{l-2}H\sigma_0, \ \ l\geq3.
\end{equation}
It follows that
\begin{equation}\notag
\begin{split}
[\phi^l\sigma_0,\sigma_0]=&(\phi^l\sigma_0)'\sigma_0=\phi'[\sigma_0]\phi^{l-1}\sigma_0+\phi(\phi^{l-1}\sigma_0)'\sigma_0\\
=&2H\phi^{l-1}\sigma_0+\phi2(l-1)\phi^{l-2}H\sigma_0\\
=&2l\phi^{l-1}H\sigma_0.
\end{split}
\end{equation}
Therefore, \eqref{b44} holds.
\end{proof}

\begin{lemma}
Introducing the notation $F_{n,m}=(\phi^{n-1}\sigma_0)'[\phi^{m}\sigma_0]$, then
\begin{equation}\label{b45}
\phi'[\phi^{n-1}\sigma_0]\phi^{m-1}\sigma_0-F_{n,m}=\phi^{n-1}\phi'[\sigma_0]\phi^{m-1}\sigma_0-\phi F_{n,m-1}, \ \ n,m=1,2,\cdots.
\end{equation}
\end{lemma}
\begin{proof}
Using \eqref{b40} repeatedly, one has
\begin{equation}\notag
\begin{split}
F_{n,m}=&\phi'[\phi^{m}\sigma_0]\phi^{n-2}\sigma_0+\phi F_{n-1,m}\\
=&\phi'[\phi^{m}\sigma_0]\phi^{n-2}\sigma_0+\phi\phi'[\phi^{m}\sigma_0]\phi^{n-3}\sigma_0+\phi^{2}F_{n-2,m}\\
=&\cdots\\
=&\sum\limits_{j=2}^n\phi^{j-2}\phi'[\phi^{m}\sigma_0]\phi^{n-j}\sigma_0.
\end{split}
\end{equation}
It follows that we obtain the following equation by using \eqref{b37} repeatedly
\begin{equation}\notag
\begin{split}
&\phi'[\phi^{n-1}\sigma_0]\phi^{m-1}\sigma_0-F_{n,m}=\phi'[\phi^{n-1}\sigma_0]\phi^{m-1}\sigma_0
-\phi'[\phi^{m}\sigma_0]\phi^{n-2}\sigma_0-\sum\limits_{j=3}^n\phi^{j-2}\phi'[\phi^{m}\sigma_0]\phi^{n-j}\sigma_0\\
&=\phi\phi'[\phi^{n-2}\sigma_0]\phi^{m-1}\sigma_0-\phi\phi'[\phi^{m-1}\sigma_0]\phi^{n-2}\sigma_0-\phi\phi'[\phi^{m}\sigma_0]\phi^{n-3}\sigma_0
-\sum\limits_{j=4}^n\phi^{j-2}\phi'[\phi^{m}\sigma_0]\phi^{n-j}\sigma_0\\
&=\cdots\\
&=-\phi F_{n,m-1}+\phi^{n-1}\phi'[\sigma_0]\phi^{m-1}\sigma_0.
\end{split}
\end{equation}
\end{proof}

\begin{lemma}
\begin{equation}\label{b46}
[\phi^m\sigma_0,\phi^n\sigma_0]=2(m-n)\phi^{m+n-1}H\sigma_0, \ \ m=1,2,\cdots,\ \ n=0,1,2,\cdots.
\end{equation}
\end{lemma}
\begin{proof}
From Lemma 5, one can find  that \eqref{b46} holds when $m=0,$ $n=0,1,2,\cdots.$
Assume \eqref{b46} holds for $m=l-1$, that is
\begin{equation}\notag
[\phi^{l-1}\sigma_0,\phi^n\sigma_0]=2(l-n-1)\phi^{l+n-2}H\sigma_0,\ \ l\geq2.
\end{equation}
In order to prove that \eqref{b46} holds, we need to prove that $[\phi^l\sigma_0,\phi^n\sigma_0]=2(l-n)\phi^{l+n-1}H\sigma_0$, which is equivalent to proving that the following recurrence relationship holds
\begin{equation}\label{b47}
\begin{split}
[\phi^l\sigma_0,\phi^n\sigma_0]=&\phi^n\phi'[\sigma_0]\phi^{l-1}\sigma_0+\phi[\phi^{l-1}\sigma_0,\phi^n\sigma_0]\\
=&\phi^n2H\phi^{l-1}\sigma_0+\phi2(l-n-1)\phi^{l+n-2}H\sigma_0=2(l-n)\phi^{l+n-1}H\sigma_0.
\end{split}
\end{equation}
Actually, the  recurrence relationship \eqref{b47} can be obtained by means of \eqref{b37}, \eqref{b40} and \eqref{b45}.
\end{proof}

According to the coupled nonisospectral KdV hierarchy \eqref{b33}, \eqref{b36} and the above lemmas, we let
\begin{equation}\label{b48}
\tau_0^m=(2m+1)tHK_{m-1}+\sigma_0, \ \ m=1,2,\cdots,
\end{equation}
\begin{equation}\label{b49}
\tau_n^m=\phi^n\tau_0^m=(2m+1)tHK_{m+n-1}+\phi^n\sigma_0, \ \ n=0,1,2,\cdots.
\end{equation}
Then, the $\tau$ symmetries of the coupled nonisospectral KdV hierarchy can be deduced.
\begin{theorem}
$\tau_n^m$ are symmetries of the coupled nonisospectral KdV hierarchy \eqref{b33}, that means
\begin{equation}\label{b50}
(\tau_n^m)_t=K_m'[\tau_n^m], \ \ m=1,2,\cdots,\ \ n=0,1,2,\cdots.
\end{equation}
\end{theorem}
\begin{proof}
From \eqref{b39} and \eqref{b41}, one has
\begin{equation}
\begin{split}
K_m'[\tau_0^m]=&K_m'((2m+1)tHK_{m-1}+\sigma_0)\\
=&(2m+1)tHK_{m-1}'[K^m]+(2m+1)HK_{m-1}\\
=&(2m+1)tHK_{m-1,t}+(2m+1)HK_{m-1}\\
=&(\tau_0^m)_t.
\end{split}
\end{equation}
Hence,  \eqref{b50} holds for $n=0.$ We conclude that \eqref{b50} holds  because $\phi$ is a strong symmetry of \eqref{b33}.
\end{proof}

\begin{theorem}
\begin{equation}\label{b51}
[K_m,\tau_n^l]=(2m+1)HK_{m+n-1}, \ \ l,m=1,2,\cdots,\ \ n=0,1,2,\cdots.
\end{equation}
\end{theorem}
\begin{proof}
By means of \eqref{b39}, \eqref{b43} and \eqref{b49}, we have
\begin{equation}\notag
[K_m,\tau_n^l]=[K_m,(2l+1)tHK_{l+n-1}+\phi^n\sigma_0]=[K_m,\phi^n\sigma_0]=(2m+1)HK_{m+n-1}.
\end{equation}
\end{proof}

\begin{theorem}
\begin{equation}\label{b52}
[\tau_l^m,\tau_n^m]=2(l-n)H\tau_{l+n-1}^m, \ \ m=1,2,\cdots,\ \ l,n=0,1,2,\cdots.
\end{equation}
\end{theorem}
\begin{proof}
From \eqref{b39}, \eqref{b43}, \eqref{b46} and \eqref{b49}, we obtain
\begin{equation}\notag
\begin{split}
[\tau_l^m,\tau_n^m]=&[(2m+1)tHK_{l+m-1}+\phi^l\sigma_0,(2m+1)tHK_{n+m-1}+\phi^n\sigma_0]\\
=&(2m+1)tH[K_{l+m-1},\phi^n\sigma_0]-(2m+1)tH[K_{n+m-1},\phi^l\sigma_0]+[\phi^l\tau_0,\phi^n\sigma_0]\\
=&(2m+1)tH((2m+2l-1)HK_{l+m+n-2}-(2m+2n-1)HK_{l+m+n-2})+2(l-n)\phi^{l+n-1}H\sigma_0\\
=&2(l-n)H((2m+1)tHK_{l+m+n-2}+\phi^{l+n-1}\sigma_0)\\
=&2(l-n)H\tau_{l+n-1}^m.
\end{split}
\end{equation}
\end{proof}

\begin{theorem}
From theorem 1, theorem 3 and theorem 4, we find that the $K$ symmetries and $\tau$ symmetries  of  \eqref{b33} constitute a set of infinite dimensional Lie algebras of the following structure:
\begin{equation}\notag
[K_m,K_n]=0, \ \  [K_m,\tau_n^l]=(2m+1)HK_{m+n-1},\ \ [\tau_l^m,\tau_n^m]=2(l-n)H\tau_{l+n-1}^m.
\end{equation}
\end{theorem}

In the following, we consider some conserved quantities  of the coupled nonisospectral KdV hierarchy \eqref{b33}. Let us first recall some basic notions and definitions (see \cite{207,208,209}).
\begin{definition}
  For a given  integrable hierarchy $u_t=K_n(u)$, $\nu$ is called the conserved covariance when it satisfies
\begin{equation}\label{b53}
\frac{d\nu}{dt}+K'^{\ast}\nu=0
\end{equation}
where $K'$ denotes the linearized operator of $K$, and $K'^{\ast}$ is a conjugate operator of $K'$.
\end{definition}
\begin{proposition}
  For a given  integrable hierarchy $u_t=K_n(u)$, if $\sigma$ is its symmetry and $\nu$ is the conserved covariance, then
\begin{equation}\notag
\int_{-\infty}^{\infty}\nu\sigma dx=<\nu,\sigma>.
\end{equation}
Actually, one has $\frac{d}{dt}<\nu,\sigma>=0$ since $<\nu,\sigma>$ is independent of time $t$.
\end{proposition}
\begin{definition}
If $F'f=<\nu,f>$,  for $\forall f\in S$, then $\nu$ is called the gradient of the functional $F$, that means $\nu=\frac{\delta F}{\delta u}$.
\end{definition}
\begin{proposition}
   If $\nu'=\nu'^{\ast}$, then $\nu$ is the gradient of the following functional
\begin{equation}\label{b54}
F=\int_{0}^{1}<\nu(\lambda u),u>d\lambda.
\end{equation}
\end{proposition}
\begin{proposition}
 If $I$ is a conserved quantity of the hierarchy $u_t=K_n(u)$, and the conserved covariance $\nu$ satisfies
\begin{equation}\notag
I'K_n=<\nu,K_n>,
\end{equation}
then
\begin{equation}\notag
\frac{\partial I}{\partial t}+<\nu,K_n>=0,
\end{equation}
that is
\begin{equation}\notag
\frac{\partial \nu}{\partial t}+K_n'^{\ast}\nu+\nu'K_n=0.
\end{equation}
Therefore, the  conserved quantities associated with the integrable hierarchy $u_t=K_n(u)$ are derived as follows:
\begin{equation}\label{b55}
I_m=\int_{0}^{1}<\partial_x^{-1}K_m(\lambda u),u>d\lambda.
\end{equation}
\end{proposition}

Based on the above definitions and propositions, we obtain  a few conserved qualities of the hierarchy \eqref{b33}
as follows:

\begin{equation}\notag
\begin{split}
I_0=\int_{0}^{1}<[\partial_x^{-1}K_0(\lambda \left(\begin{matrix}
u_1\\
u_2\end{matrix}
\right))]^T,\left(\begin{matrix}
u_1\\
u_2\end{matrix}
\right)>d\lambda
=\int_{0}^{1}\int_{-\infty}^{\infty}\lambda(u_1^2+u_2^2)dxd\lambda
=\frac{1}{2}\int_{-\infty}^{\infty}(u_1^2+u_2^2)dx,
\end{split}
\end{equation}
\begin{equation}\notag
\begin{split}
I_1=&\int_{0}^{1}<\partial_x^{-1}K_1(\lambda \left(\begin{matrix}
u_1\\
u_2\end{matrix}
\right)),\left(\begin{matrix}
u_1\\
u_2\end{matrix}
\right)>d\lambda=\int_{-\infty}^{\infty}(\frac{1}{2}u_1u_{1,xx}+\frac{1}{2}u_2u_{2,xx}+u_1^3+\varepsilon u_1u_2^2+2u_1u_2^2)dx\\
=&\int_{-\infty}^{\infty}(-\frac{1}{2}u_{1x}^2-\frac{1}{2}u_{2x}^2+u_1^3+\varepsilon u_1u_2^2+2u_1u_2^2)dx,
\end{split}
\end{equation}
\begin{equation}\notag
\cdots
\end{equation}
where $K_0$ and $K_1$ are given by \eqref{b36}.

\section{A multi-component nonisospectral KdV hierarchy}

For the  Lie algebra $A_{1N}$ (see eq.\eqref{6}), we consider the corresponding loop algebra
\begin{equation}\notag
\widetilde{A}_{1N}=\rm span\{\widetilde{h}_1(n), \widetilde{h}_2(n), \widetilde{h}_3(n),\cdots, \widetilde{h}_{3N-2}(n), \widetilde{h}_{3N-1}(n), \widetilde{h}_{3N}(n)\}
\end{equation}
where
\begin{equation}\notag
\widetilde{h}_i(n)=\widetilde{h}_i\lambda^{n}, \ i=1,2,\cdots,3N.
\end{equation}
Introducing the following $N$-dimensional nonisospectral problem based on $\widetilde{A}_{1N}$
\begin{equation}\label{c12}
\begin{cases}
\psi_x=\overline{U}\psi,\ \ \overline{U}=\frac{1}{4}\widetilde{h}_2(1)+\widetilde{h}_3(0)-\sum\limits_{i=1}^{N}u_i\widetilde{h}_{3i-1}(0)=\left[
\begin{matrix}
\overline{U}_1,\overline{U}_2,\cdots,\overline{U}_N
\end{matrix}
\right]^T, \\
\psi_t=\overline{W}\psi,\ \ \overline{W}=\sum\limits_{i=1}^{N}(a_i\widetilde{h}_{3i-2}(0)+b_i\widetilde{h}_{3i-1}(0)+c_i\widetilde{h}_{3i}(0))=\left[
\begin{matrix}
\overline{W}_1,\overline{W}_2,\cdots,\overline{W}_N
\end{matrix}
\right]^T,\\
\lambda_t=\sum\limits_{m\geq0}k_m(t)\lambda^{-m},
\end{cases}
\end{equation}
where
 \begin{equation}\notag
 \begin{split}
 &\overline{U}_1=\left(\begin{matrix}
0&-u_1+\frac{\lambda}{4}\\
1&0
\end{matrix}
\right),\ \ U_l=\left(\begin{matrix}
0&-u_l\\
0&0
\end{matrix}
\right), \ \ l=2,3,\cdots,N, \ \ \overline{W}_k=\left(\begin{matrix}
a_k&b_k\\
c_k&-a_k\end{matrix}
\right),\\
 &a_k=\sum\limits_{m\geq0}a_{km}\lambda^{-m},\ \  b_k=\sum\limits_{m\geq0}b_{km}\lambda^{-m}, \ \ c_k=\sum\limits_{m\geq0}c_{km}\lambda^{-m},\ \  \ k=1,2,\cdots,N.
  \end{split}
 \end{equation}
The   stationary zero-curvature equation
\begin{equation}\label{c13}
\overline{W}_x=\frac{\partial \overline{U}}{\partial \lambda}\lambda_t+[\overline{U},\overline{W}],
\end{equation}
gives rise to the  recursion equations:
\begin{equation}\label{c14}
\begin{cases}
a_{kl,x}=\frac{1}{4}c_{k,l+1}-b_{k,l}-\sum\limits_{i+j=k+1\atop1\leq i,j\leq k}c_{il}u_j-\sigma\varepsilon\sum\limits_{m+n=k+N+1\atop k+1\leq m,n\leq N} c_{ml}u_n,\\
b_{kl,x}=\frac{1}{4}k_l(t)-\frac{1}{2}a_{k,l+1}+2\sum\limits_{i+j=k+1\atop1\leq i,j\leq k}a_{il}u_j+2\sigma\varepsilon\sum\limits_{m+n=k+N+1\atop k+1\leq m,n\leq N} a_{ml}u_n,\\
c_{kl,x}=2a_{kl},\ \ \ l\geq0, \ \  \ k=1,2,\cdots,N,
\end{cases}
\end{equation}
which are  equivalent to
\begin{equation}\label{c15}
\begin{cases}
a_{kl}=\frac{1}{2}c_{kl,x},\ \ \ l\geq0, \ \  \ k=1,2,\cdots,N,\\
b_{kl}=-\frac{1}{4}c_{k,l+1}+\partial^{-1}(\sum\limits_{i+j=k+1\atop1\leq i,j\leq k}c_{il,x}u_j+\sigma\varepsilon\sum\limits_{m+n=k+N+1\atop k+1\leq m,n\leq N} c_{ml,x}u_n)+\frac{1}{4}k_l(t)x,\\
c_{k,l+1}=c_{kl,xx}+2\sum\limits_{i+j=k+1\atop1\leq i,j\leq k}c_{il}u_j+2\sigma\varepsilon\sum\limits_{m+n=k+N+1\atop k+1\leq m,n\leq N} c_{ml}u_n\\
\ \ \ \ \ \ \ \ \ \ \ +2\partial^{-1}(\sum\limits_{i+j=k+1\atop1\leq i,j\leq k}c_{il,x}u_j+\sigma\varepsilon\sum\limits_{m+n=k+N+1\atop k+1\leq m,n\leq N} c_{ml,x}u_n)+\frac{1}{2}k_l(t)x.
\end{cases}
\end{equation}
To the  recursion equations \eqref{c15}, we take the initial values $ a_{k0}=0,$
it follows that one has
\begin{equation}\label{c15*}
\begin{split}
a_{k0}=&0,\ \ c_{k0}=\beta_1,\ \ c_{k1}=2\beta_1\sum\limits_{j=1}^ku_j+2\beta_1\sigma\varepsilon\sum\limits_{n=k+1}^Nu_n+\frac{1}{2}k_0(t)x,\\
b_{k0}=&-\frac{\beta_1}{2}\sum\limits_{j=1}^ku_j-\frac{\beta_1}{2}\sigma\varepsilon\sum\limits_{n=k+1}^Nu_n+\frac{1}{8}k_0(t)x,\ \ a_{k1}=\beta_1\sum\limits_{j=1}^ku_{j,x}+\beta_1\sigma\varepsilon\sum\limits_{n=k+1}^Nu_{n,x}+\frac{1}{4}k_0(t),\\
c_{k2}=&2\beta_1u_{k,xx}+6\beta_1(\sum\limits_{i+j=k+1\atop1\leq i,j\leq k}u_iu_j+\sigma\varepsilon\sum\limits_{m+n=k+N+1\atop k+1\leq m,n\leq N} u_mu_n)\\
&+k_0(t)(x+\partial^{-1})(\sum\limits_{j=1}^ku_j+\sigma\varepsilon\sum\limits_{n=k+1}^Nu_n)+\frac{1}{2}k_1(t)x,
\\
&\cdots
\end{split}
\end{equation}

Denoting that
\begin{equation}\notag
\begin{split}
&\overline{W}^{(n)}=\overline{W}\lambda^n=\sum\limits_{i=1}^{N}(a_i\widetilde{h}_{3i-2}(n)+b_i\widetilde{h}_{3i-1}(n)+c_i\widetilde{h}_{3i}(n))
=\overline{W}_+^{(n)}+\overline{W}_-^{(n)}, \\
&\overline{W}_+^{(n)}=\sum\limits_{i=1}^{N}\sum\limits_{m=0}^{n}(a_{im}\widetilde{h}_{3i-2}(0)+b_{im}\widetilde{h}_{3i-1}(0)+c_{im}\widetilde{h}_{3i}(0))\lambda^{n-m}.
\end{split}
\end{equation}
Taking the modified term  $\triangle_n=-\frac{1}{4}\sum\limits_{i=1}^{N}c_{i,n+1}\widetilde{h}_{3i-1}(0)$ so that for $\overline{V}^{(n)}=\overline{W}_+^{(n)}+\triangle_n$.
It follows that the nonisospectral zero curvature equation
\begin{equation}\label{c18}
\frac{\partial \overline{U}}{\partial u}u_t+\frac{\partial \overline{U}}{\partial \lambda}\lambda_{t}^{(n)}-\overline{V}_x^{(n)}+[\overline{U},\overline{V}^{(n)}]=0,
\end{equation}
gives rise to the  multi-component nonisospectral KdV hierarchy  as follows:
\begin{equation}\label{c19}
\begin{split}
u_{t_n}&=\overline{K}_n=\frac{1}{2}\partial\left(\begin{matrix}
c_{1,n+1}\\
c_{2,n+1}\\
\vdots\\
c_{N,n+1}
\end{matrix}
\right)=J\left(\begin{matrix}
c_{1,n+1}\\
c_{2,n+1}\\
\vdots\\
c_{N,n+1}
\end{matrix}
\right)
=J[\widetilde{L}\left(\begin{matrix}
c_{1,n}\\
c_{2,n}\\
\vdots\\
c_{N,n}
\end{matrix}
\right)+\frac{1}{2}k_n(t)\left(\begin{matrix}
x\\
x\\
\vdots\\
x
\end{matrix}
\right)]\\
&=:\widetilde{M}\left(\begin{matrix}
c_{1,n}\\
c_{2,n}\\
\vdots\\
c_{N,n}
\end{matrix}
\right)+\frac{1}{4}k_n(t)\left(\begin{matrix}
1\\
1\\
\vdots\\
1
\end{matrix}
\right)=\widetilde{\Phi}^nJ\left(\begin{matrix}
c_{1,1}\\
c_{2,1}\\
\vdots\\
c_{N,1}
\end{matrix}
\right)+\frac{1}{2}\sum\limits_{m=1}^{n}k_{m}(t)\widetilde{\Phi}^{n-m}\left(\begin{matrix}
\frac{1}{2}\\
\frac{1}{2}\\
\vdots\\
\frac{1}{2}
\end{matrix}
\right),
\end{split}
\end{equation}
where $J$ is given by \eqref{25} and $\widetilde{L}$  is determined by the  recursion equations \eqref{c15}
\begin{equation}\notag
\begin{split}
& \widetilde{L}=\left[
\begin{matrix}
\widetilde{L}_1,\widetilde{L}_2,\cdots,\widetilde{L}_N,
\end{matrix}
\right]^T,\  \ \widetilde{L}_1=\partial^2+4u_1-2\partial^{-1}u_{1x},\ \ \widetilde{L}_m=4u_m-2\partial^{-1}u_{mx},\ \ m=2,\cdots,N,\\
&\widetilde{M}=J\widetilde{L}=\frac{1}{2}\left[
\begin{matrix}
\widetilde{M}_1,\widetilde{M}_2,\cdots,\widetilde{M}_N,
\end{matrix}
\right]^T,\  \ \widetilde{M}_1=\partial^3+2u_1\partial+2\partial u_{1},\ \ \widetilde{M}_m=2u_m\partial+2\partial u_{m},\ \ m=2,\cdots,N,\\
&\widetilde{\Phi}=J\widetilde{L}J^{-1}=\left[
\begin{matrix}
\widetilde{\Phi}_1,\widetilde{\Phi}_2,\cdots,\widetilde{\Phi}_N,
\end{matrix}
\right]^T,\  \ \widetilde{\Phi}_1=\partial^2+2u_{1,x}\partial^{-1}+4u_{1},\ \ \widetilde{\Phi}_m=2u_{m,x}\partial^{-1}+4u_{m},\ \ m=2,\cdots,N.
\end{split}
\end{equation}
The first two   examples in the above  hierarchy of soliton equations are
\begin{equation}\label{c21}
\begin{split}
u_{t_0}=\left(\begin{matrix}
u_1\\
u_2\\
\vdots\\
u_N
\end{matrix}
\right)_{t_0}=\frac{1}{2}\partial\left(\begin{matrix}
c_{1,1}\\
c_{2,1}\\
\vdots\\
c_{N,1}
\end{matrix}
\right),
\end{split}
\end{equation}
\begin{equation}\notag
u_{k,t_0}=\frac{1}{2}(c_{k,1})_x=\beta_1\sum\limits_{j=1}^ku_{j,x}+\beta_1\sigma\varepsilon\sum\limits_{n=k+1}^Nu_{n,x}+\frac{1}{4}k_0(t),\ \ \ k=1,\cdots,N.
\end{equation}
\begin{equation}\label{c22}
\begin{split}
u_{t_1}=\left(\begin{matrix}
u_1\\
u_2\\
\vdots\\
u_N
\end{matrix}
\right)_{t_1}=\frac{1}{2}\partial\left(\begin{matrix}
c_{1,2}\\
c_{2,2}\\
\vdots\\
c_{N,2}
\end{matrix}
\right),
\end{split}
\end{equation}
\begin{equation}\notag
\begin{split}
u_{k,t_1}=&\frac{1}{2}(c_{k,2})_x=\beta_1u_{k,xxx}+3\beta_1(\sum\limits_{i+j=k+1\atop1\leq i,j\leq k}u_iu_j+\sigma\varepsilon\sum\limits_{m+n=k+N+1\atop k+1\leq m,n\leq N} u_mu_n)_x\\
&+\frac{1}{2}k_0(t)[2(\sum\limits_{j=1}^ku_j+\sigma\varepsilon\sum\limits_{n=k+1}^Nu_n)+x(\sum\limits_{j=1}^ku_{j,x}
+\sigma\varepsilon\sum\limits_{n=k+1}^Nu_{n,x})]+\frac{1}{4}k_1(t),\ \ k=1,\cdots,N,
\end{split}
\end{equation}
which is a multi-component nonisospectral KdV equation.\\
 When $N=1,2$, the system \eqref{c22} is reduced to the nonisospectral KdV equation \eqref{22} and  the coupled nonisospectral KdV equation \eqref{b22} respectively.

\section{Conclusions and discussions}

The higher-dimensional Lie algebras \eqref{3}-\eqref{11} were constructed.  As the applications, we considered the nonisospectral problems \eqref{12}, \eqref{b12}, \eqref{c12} respectively, and thus deduced the  nonisospectral KdV hierarchy \eqref{19}, the  coupled nonisospectral KdV hierarchy  \eqref{b19} and
the  multi-component nonisospectral KdV hierarchy \eqref{c19}. By reducing these hierarchies, we obtained the famous KdV equation \eqref{22}, coupled nonisospectral KdV equation \eqref{b22} and multi-component nonisospectral KdV equation \eqref{c22} respectively. It follows that the  bi-Hamiltonian structures of these resulting hierarchies  were derived by means of the Tu scheme and thus showing their Liouville integrability. Additionally, we found that the $K$ symmetries and $\tau$ symmetries of the coupled nonisospectral KdV hierarchy constitute a set of infinite dimensional Lie algebras.

This paper only considers the KdV space spectral problem.  In fact, the idea and the method used here are universal for other isospectral and nonisospectral problems. The Riemann-Hilbert method  for multi-component systems based on higher-order matrix spectral problems  had been discussed (see \cite{215,216,217}). It is quite intriguing for us  to  consider the application of the Riemann-Hilbert approach to
the multi-component KdV systems. Also, the application of $\overline{\partial}$-dressing method  for the construction of solutions to the  multi-component KdV systems (see \cite{218}), which is worthy of further work.

{\bf Declarations}

{\bf Ethics approval and consent to participate} All authors approve ethics and consent to participate.

{\bf Consent for publication} All authors consent for publication.

{\bf Availability of data and materials} In this paper, we have no data and materials used.

{\bf Competing interests} The authors declare that they have no conflicts of interest.

{\bf Funding} This work was supported by the Natural Science Foundation of Fujian Province of China (grant No.2024J01724) and the National Natural Science Foundation of China (grant No.12371256).

{\bf Authors' contributions} HW: Formal analysis, Writing-original draft, Writing-review $\&$ editing. YZ: Investigation, Supervision, Funding
acquisition. BF: Project administration, Software, Data
curation. All authors read and approved
the fnal manuscript. All the authors contributed equally to the writing of this paper.

{\bf Acknowledgements} The authors would like to thank the referee for valuable comments and suggestions
on this article.


\end{CJK*}
\end{document}